\documentclass[reprint, showpacs, nofootinbib]{revtex4-1}

\usepackage{amsmath,amsthm,amssymb,amsfonts}

\usepackage{mathrsfs}
\usepackage{cases}
\usepackage{graphicx}
\usepackage{enumerate}
\usepackage{hyperref}

\numberwithin{equation}{section}
\newcommand{\ii}{{\rm i}}
\newcommand{\ee}{{\rm e}}
\newcommand{\x}{{\rm x}}
\newcommand{\y}{{\rm y}}
\newcommand{\z}{{\rm z}}
\newcommand{\I}{1\!\!1}
\newcommand{\vol}{{\rm vol}}
\newcommand{\ad}{{\rm ad}}

\newtheorem{thm}{Theorem}
\newtheorem{lemma}[thm]{Lemma}

\theoremstyle{definition}
\newtheorem{defn}[thm]{Definition}

\begin{document}

\title{Generally covariant dynamical reduction models and the Hadamard condition}

\author{Benito A. Ju\'arez-Aubry}
\email{benito.juarez@correo.nucleares.unam.mx}
\affiliation{Departamento de Gravitaci\'on y Teor\'ia de Campos \\Instituto  de  Ciencias  Nucleares,  Universidad  Nacional Aut\'onoma de M\'exico, \\A. Postal 70-543, Mexico City 045010, Mexico.
}

\author{Bernard S. Kay}
\email{bernard.kay@york.ac.uk}
\affiliation{Department of Mathematics, University of York, York YO10 5DD, UK.}

\author{Daniel Sudarsky}
\email{daniel.sudarsky@nucleares.unam.mx}
\affiliation{Departamento de Gravitaci\'on y Teor\'ia de Campos \\Instituto  de  Ciencias  Nucleares,  Universidad  Nacional Aut\'onoma de M\'exico, \\A. Postal 70-543, Mexico City 045010, Mexico. \vspace*{0.5cm}
}

\date{V2: \today}

\begin{abstract} 

This is an author-created, un-copyedited version of an article published in Phys.Rev. {\bf D} 97 (2018) no.2, 025010. Errors and ommissions in this version are the sole responsability of the authors. The published version is available online at DOI:10.1103/PhysRevD.97.025010. \\

We provide general guidelines for generalizing dynamical reduction models to curved spacetimes and propose a class of generally covariant relativistic versions of the GRW model. We anticipate that the collapse operators of our class of models may play a r\^ole in a yet-to-be-formulated theory of semiclassical gravity with collapses. We show explicitly that the collapse operators map a  dense domain of states that are initially Hadamard to final Hadamard states --  a property that we expect will be needed for the construction of such a semiclassical theory. Finally, we provide a simple example in which we explicitly compute the violations in energy-momentum due to the state reduction process and conclude that this violation is of the order of a parameter of the model -- supposed to be small. 
\end{abstract}


\maketitle

%=====================================================================================================
% INTRODUCTION
%=====================================================================================================
\section{Introduction}
\label{sec:Intro}

Quantum theory, despite being empirically extraordinarily successful, continues to be beset by the so-called \textit{measurement problem}.   In the Schr\"odinger picture, the quantum state of an undisturbed system is supposed to evolve unitarily according to the Schr\"odinger equation, $\dot\psi = -\ii \widehat{H}\psi$, where $\widehat H$ is the quantum Hamiltonian. Yet, when the measurement of an observable $A$ is performed, the state is supposed to undergo a process called reduction\footnote{\label{ftnt1} Here we take $\widehat{A}$ to have nondegenerate discrete spectrum.   von Neumann \cite{Neumann} calls the unitary and reduction processes Rule II and Rule I;  below we shall follow Penrose \cite{Penrose1, Penrose1a, Penrose:1996cv, Penrose2} in calling them $U$ and $R$ respectively.} whereby it changes abruptly from the pre-measurement state to an eigenstate of the operator $\widehat{A}$, $\psi\mapsto \alpha_i$, where $i$ is an index set element, with probability $|\langle \psi | \alpha_i \rangle|^2$, where $\alpha_i$ is an eigenvector of $\widehat{A}$ with eigenvalue $a_i$, which is in turn taken to be the outcome of the measurement.

As already discussed by the founders of the subject (see the account in \cite{WheelerZurek}) culminating in von Neumann's book \cite{Neumann} and in numerous subsequent textbooks (of which we mention for example \cite{Jauch, Jammer, Peres, Omnes, BuschEtAl, Landsman} -- see also the collection of articles by Bell \cite{Bell} and also \cite{Bassi:2003gd}, to which we also refer later for collapse models) the measurement problem has its origin in the fact that the theory does not specify what a measurement is and therefore it is never completely clear which of the two rules, $U$ or $R$ (see footnote \ref{ftnt1}), should be applied in any particular situation.  If, in a situation that might be considered a measurement, one chooses to model the measurement apparatus as a second quantum system coupled to the measured quantum system and applies $U$, one predicts the existence of (macroscopic) superposition states, in which the apparatus is entangled with the measured system, which have no counterpart when one, instead, applies $R$.    Although, as explained long ago by Heisenberg and by von Neumann (in \cite{Neumann}), the final prediction will not differ significantly for most practical purposes, there remains an unsatisfactory vagueness (see, {\it e.g.}, Bell's account \cite{BellAgainst}).  Also, the seeming presence, on the former choice, of macroscopic superposition states (as illustrated by Schr\"odinger's cat) is troubling and seems to be at odds with our classical understanding of macroscopic systems.  Furthermore   in  certain   contexts,   such as  cosmology  and black hole physics,  the problem cannot be easily bypassed.

Here we focus on some specific technical issues within one set of proposals for resolving some aspects of the measurement problem. Namely the so-called {\it (spontaneous) collapse} or {\it dynamical reduction} models.  

The first suggestion of a  dynamics for wave function collapse was by Bohm and Bub 
\cite{BOHM:1966zz}.  This was followed by a specific proposal to describe the collapse as a dynamical process by Pearle in \cite{Pearle:1976ka} and  \cite{Pearle:1979vm}, which  however  faced   the so-called  `trigger problem' and the `preferred basis problem'.    These   were successfully   resolved  in  the proposal by Ghirardi, Rimini and Weber \cite{Ghirardi:1985mt}.

This proposal attracted the attention of Bell who, in
\cite{BellJumps}, formulated it in terms of a stochastic time-evolving wave function\footnote{In \cite{Ghirardi:1985mt}, the theory was formulated in terms of a deterministic time-evolving density operator, $\widehat{\rho}(t)$.   In modern terminology what Bell supplied was one particular `unraveling' of the master equation that determines $\widehat{\rho}(t)$.} (See also \cite{Ghirardi:1987ns}.) Bell also argued that the theory contains enough to remove the concerns that a relativistic collapse theory may be impossible. As is customary, we shall refer to the theory of {\cite{Ghirardi:1985mt, BellJumps} as the GRW theory.

A major issue with the original GRW collapse theory is that it does not incorporate the Bose-Einstein (or Fermi-Dirac symmetry) (or antisymmetry) needed to deal with identical particles.  This  was   first  fixed  in  the  CSL  (continuous spontaneous localization) model \cite{Pearle:1988uh},\footnote{As emphasized in the recent paper \cite{NewGR}, the difference between the GRW and CSL models should, from a physical point of view, be regarded as a relatively minor technicality; as is explained there, what really matters is the choice of which observables are ``made sharp'', or, in the language we use in the sequel here, of which are the relevant {\it collapse generators}.} and independently by Di\'osi in \cite{Diosi:1988uy}.

An important result from that period, discussed in \cite{GisinResponds}, is a  condition on the viable collapse  models. The point is that the time-evolving density operator characterizing the  modified   evolution of  a  statistical ensemble of systems must be determined by a master equation of the GKS-Lindblad form \cite{GKS, Lindblad} if the model is to avoid the possibility of superluminal communication. Once possessing a satisfactory collapse dynamics, it is sometimes useful to restrict attention just to the associated  GKS-Lindblad master equation   for the  corresponding time-evolving density operator.  We refer to this as the \textit{density-operator formulation of the theory}. 

In fact as initially  shown in \cite{Ghirardi:1989cn}, given  any GKS-Lindblad master equation for a  time-evolving density operator $\widehat{\rho}(t)$, it is possible to find\footnote{Strictly, as explained in \cite{Ghirardi:1989cn} this statement holds within the CSL formalism and not always in the GRW formalism.}  a stochastic dynamical rule determining a  stochastically time-evolving wave function, $\psi_s(t)$ -- $s$ standing for the relevant stochastic parameters -- such that $\widehat{\rho}(t)$ is the average over the ensemble labelled by $s$ of $|\psi_s(t)\rangle\langle\psi_s(t)|$.  The stochastic dynamical rule that determines such an ensemble of time-evolving wave functions is called an \textit{unraveling} of the time-evolving density operator or equivalently of the master equation that determines it.  (The notion is due originally to Carmichael \cite{Carmichael}.)  Sometimes, one also refers to $\psi_s(t)$ as an unraveling of $\widehat{\rho}(t)$. However, a crucial point is that  there are, in general, multiple  possibile unravelings of the same GKS-Lindblad  equation,   and, in the current work,  we  consider  the formulation of the theory in  terms of a stochastically time-evolving wave vector to be more fundamental.  We will refer to this as the \textit{wave-vector formulation of the theory}.

More modern developments include substantial experimental programs.  For a
recent review,  see \cite{Bassi:2012bg}.

In relation to the measurement problem that we discussed at the outset, the main positive feature of the GRW and CSL models is that they replace $U$ and $R$ by an objective set of (stochastic) rules that mimics a particular combination of the applications of $U$ and $R$ and that one can show that, in many cases has the effect of eliminating the troubling macroscopic superposition states.   

These positive features of collapse models come at a price:  Dynamical collapse models are not without conceptual difficulties of their own, such as, for example, the `tail problem' but (see, {\it e.g.}, \cite{Wallace}) these seem to be resolvable.    Moreover the rules involve certain parameters,  functions, and stochastic processes, which are partly fixed by pragmatically tuning them so as not to conflict with any known phenomena, but which retain a high degree of arbitrariness.  However, our attitude to these models should perhaps be that they are just stopgaps which will one day be replaced by a more fundamental theory in which the several presently \textit{ad hoc} and partly arbitrary parameters and functions become calculable in terms of existing fundamental constants.   In particular there are reasons to think that quantum gravity, when understood better than we presently do, will do this job  -- perhaps along the lines adumbrated by Penrose in \cite{Penrose1, Penrose1a, Penrose:1996cv, Penrose2} or by Di\'osi in \cite{Diosi:1988uy, DiosiQGC} and/or perhaps in line with the `matter-gravity entanglement hypothesis' of one of us (see \cite{Kay1998a, Kay1998b, KayAbyaneh, Kay:2015csa} and references therein). In the meantime, by studying the predictions of existing dynamical reduction models, we hope to be able to learn lessons and make testable predictions which may one day be confirmed by such more fundamental theories. 

An important drawback of early dynamical reduction models is that they are non-relativistic, but the development of these early models quickly led to enquiries   about  their possible relativistic  generalizations.

A  valuable   concrete   proof of  existence of  collapse  models   compatible  with   special relativity is provided  by \cite{Tumulka2006}, although   that specific   example  only deals  with   situations  involving a  fixed  finite  number of non-interacting quantum particles, and it is not  clear how  it  might   be generalized to fully  quantum field theoretical settings.

Earlier   considerations concerning the     general requirements  such theories  must possess appeared in \cite{Pearle62, Grassi}. The constructive exchanges  in \cite{Aharonov:1984zzb, MyrvoldBecoming, MyrvoldPeaceful}, together with the  introduction  of an auxiliary ``pointer  field" into the dynamical reduction models  in \cite{Pearle:2005, Pearle:2008} eventually lead to  the development  of special   relativistic  versions  of  collapse theories  fully adapted to the context  quantum fields  \cite{Bedingham:2010hz, Pearle:2014tda}. In fact a recent work  \cite{MyrvoldDegFree} argues  that    viable  relativistically covariant collapse theories  must make use of such  kind of non-standard fields as the pointer field, alluded to above. In this regard,  we should  point out that  the  possibility of having the   collapse dynamics  tied to the   curvature of  spacetime,  as    is  considered in \cite{Modak:2014qja, Modak:2014vya, Bedingham:2016aus, Modak:2016uwr}, might  allow  one  to bypass  such a conclusion.  

We remark that the inherent non-locality that  must be  present in these  models, and discussed  in \cite{Aharonov:1984zzb, MyrvoldBecoming, MyrvoldPeaceful},  is  nevertheless    safe   regarding   faster than light communication ({\it i.e.}  the models  do not allow  it).

From a philosophical standpoint, addressing the tension between the locality of special and general relativity and the nonlocal aspects of global quantum states is relevant for dynamical reduction models, and a range of positions appear in the literature. 

We  might  be motivated by Penrose's work \cite[p. 446]{Penrose1} to contemplate  a radical revision of special relativity as being possibly  necessary before quantum collapse can be made consistent with relativity. On the other side, Kochen \cite{Kochen:2017gay} has recently argued that, when the concepts of quantum theory are appropriately conceptualized, the theory contains no nonlocal features. 

To exemplify the latter posture, one could be led to argue that the measurement of a property of one subsystem of an EPR pair, say the spin along the z-direction of one of two spacelike separated particles, should not be understood as a measurement on the other subsystem. This position does away with the need for essential (non-epistemic) nonlocalities. A difficulty with such a sort of instrumentalist viewpoint is that it fails to define what a measurement is. For example, it does not explain why in the above example one should not consider a measurement of the spin of particle one also as a measurement on particle two, while the  observation of  a black  dot on the screen  of a  Stern-Geralch experiment presumably does  provide a   measurement of the   corresponding spin component of the   quantum particle. The only difference we can see is that in the  latter case locality would not be compromised, while  it is in the former  it clearly would be.

The position that we take is symptathetic to the posture of Myrvold \cite{Myrvold} and Maudlin \cite{Maudlin}, which argues for a clear distinction between ``action at a distance" and nonlocality. The point is that dynamical reduction models can incorporate   nonlocal aspects of physics in the form of, say, nonlocal correlations without the need to ascribe an  asymmetric causal relation to the spatially-correlated outcomes of experiments, such asymmetry being the only thing that is forbidden by Lorentz invariance. We  point the reader to the above cited references for in-depth  discussions of these conceptual subtleties, which, after all, are not the central focus of the present paper.

From a physics standpoint, a serious complication for special relativistic models is that any theory that produced particles out of the Minkowski vacuum state would lead to infinite particle creation, due to the Poincar\'e invariance of the vacuum. This obstruction was tamed by techniques due to Pearle and Bedingham, with the aid of the aforementioned pointer field (see \cite{Pearle:2014tda} and \cite{Bedingham:2010hz} respectively). On the other hand, in \cite{Tumulka:2005ki} Tumulka takes the approach that a flash ontology of collapse theories might lead to a Lorentz-invariant formulation of dynamical reduction models (without the aid of external auxiliary fields) but, while some progress is made in this direction, a concrete relativistic model is not formulated explicitly in \cite{Tumulka:2005ki}.

In the aforementioned works \cite{Pearle:2014tda, Bedingham:2010hz, Tumulka:2005ki}, a fundamental ingredient in the construction of the models is the choice of a number operator (smeared as a number density or mass density operator) acting on the Fock space of the quantum field theory. Hence, from the outset those formulations rely on a field representation and the choice of a vacuum state, which are structures that are unnatural in curved spacetimes in the absence of symmetries. Not only that, but even in Minkowski space one might, {\it e.g.}, decide to count Fulling-Rindler particles rather than Minkowski particles.

A particularly relevant issue to us here is that dynamical reduction models generically violate conservation of energy, on which we remark the parameters of each model can be chosen to be consistent with the phenomenological constraints on energy violation (along with the constraints on all other phenomena).

When moving further to a general relativistic framework, as in the work of Bedingham et al.\ \cite{Bedingham:2016aus}, which generalizes \cite{Bedingham:2010hz} to curved spacetimes, one expects that the theory will violate the conservation of the renormalized energy-momentum tensor, but if one wants a modified version of  the semiclassical Einstein equations\footnote{By (unmodified) semiclassical Einstein equations we mean a replacement for the classical Einstein equations, $G_{ab}=8\pi G_\text{N}T_{ab}$ in which the classical energy momentum tensor gets replaced by the expectation value, $\langle T_{ab}\rangle$, of a renormalized quantum energy momentum tensor in a suitable (Heisenberg) state.} to hold, one will at least need the theory be defined so that an initial state, for which the expectation value of the renormalized energy-momentum tensor is finite, avoids evolving, upon collapse, into a state for which it has an infinite expectation value.   In the case of a linear scalar quantum field model, to which we shall restrict our attention, this essentially reduces to needing the property that any given initial Hadamard state will evolve dynamically into a state which is also Hadamard.\footnote{Actually, this property can be relaxed somewhat:  In Section \ref{sec:Result}, we in fact content ourselves with showing that, for a suitably large class of Hadamard states, a state in this class will evolve dynamically into another state in this class -- and will therefore, in particular, itself be Hadamard.}

Here, we recall that, in QFT in CS (see \cite{KayEncyc} for a brief general introduction to the topic), the conventional wisdom  is that there does not exist a preferred vacuum state for the theory in the absence of special spacetime symmetries or asymptotic behaviour.   Instead, for linear field theories, there exists a class of physically admissible states, called {\it Hadamard states}, for which the field two-point function has an appropriate short-distance behaviour so that the expectation value of the renormalized energy-momentum tensor of the theory is finite.   The earliest formulations of the Hadamard property culminated in the rigorous definition provided by one of us and Wald \cite{Kay:1988mu} and it is this definition that we shall use throughout this paper. (For more about the Hadamard condition and the renormalized energy-momentum tensor see, {\it e.g.}, the articles \cite{Wald78, Decanini:2005gt} and the standard monographs \cite{Birrell:1982ix, Wald:1995yp}.)

While, in the present paper, our main interest in the Hadamard condition is that it is essentially equivalent to the existence of a finite renormalized energy-momentum tensor, let us point out, in passing, that several interesting results have recently been proven about Hadamard states.  Most of these rely on a more recent, equivalent and very powerful, alternative definition of the Hadamard condition, now known as the {\it microlocal spectrum condition}, which was provided by Radzikowski \cite{Radzikowski:1996pa, KhavkineMoretti} in 1996 using notions from microlocal analysis and which involves a generalization of the spectrum condition (see, {\it e.g.}, \cite{Haag:1992hx}) in flat-spacetime QFT.   Amongst these is the result \cite{Fewster:2011pe} that, under certain QFT structural conditions, which are satisfied, {\it e.g.}, by the Klein-Gordon theory \cite{Fewster:2011pn}, then (in accordance with the conventional wisdom mentioned above) there exists no prescription for defining a state which is covariant in a suitable sense (for which, again, see \cite{Fewster:2011pe} and references therein) under changes from one spacetime to another and which has the Hadamard property in all of them.   This result helps us to understand (see \cite{Fewster:2013lqa}) why attempts (see, {\it e.g.}, the recent paper  \cite{Afshordi:2012jf}) to give general constructions for preferred states in general spacetimes necessarily fail to have the Hadamard short-distance behaviour, and thereby yield unphysical results.  The Hadamard property is also crucial in the definition of local Wick polynomials and time-ordered products in QFT in CS \cite{Hollands:2001nf, Hollands:2001fb}, as well as in the definition of particles in general curved spacetimes situations \cite{Louko:2007mu, Juarez-Aubry:2014jba, Hodgkinson:2011pc}.

We wish now to clarify the main purpose of the present paper as being to answer the question whether, for a linear scalar field theory model (namely, the Klein-Gordon theory), there can exist a generally covariant quantum field theoretic model that fits into Bedingham {\it et al.}'s framework for QFT in CS with collapses with the property that suitable initial Hadamard states evolve to Hadamard states at later times. We will answer this question in the affirmative.   In the course of doing so, we will give general guidelines towards the construction of such theories.

We consider that the results of this paper are relevant towards putting relativistic dynamical reduction models on a firm theoretical foudation, and can {\it en passant} also lead to further progress in the line of work \cite{Perez:2005gh, Bengochea:2014pka, Leon:2015eza, Landau:2012a, Josset:2016vrq}, which is aimed at understanding the formation of structure in quantum cosmological (inflationary) models, and also in the line of work \cite{Bedingham:2016aus, Okon:2014dpa, Modak:2014qja, Modak:2014vya}, in the physics of quantum black holes, offering a possible resolution of the information loss puzzle \cite{HawkingInfoLoss}. Additionally,  having dynamical reduction models, defined by {\it collapse operators} that map initial Hadamard states to final Hadamard states, would seem to be a prerequisite for a viable semiclassical theory of gravity that incorporates collapses.

The plan of the paper is as follows: In section \ref{sec:DynRed}, after a brief review of the necessary QFT in CS formalism, we review Bedingham's framework for a relativistic generalization of the GRW  and  CSL  models and discuss some lines for their further development as generalized to curved spacetimes in \cite{Bedingham:2016aus}; in particular, from a set of physically motivated conditions, we arrive at a  generalization of the GRW model to a generally covariant QFT in CS model.   

In section \ref{sec:Hadamard} we prepare the ground for the main result of the paper by recalling the necessary details from the definition of the Hadamard condition in \cite{Kay:1988mu} and how Hadamard states yield a finite value for the expectation value of the renormalized energy-momentum tensor. 

In section \ref{sec:Result} we then prove the main results of this paper. Namely, that there exist certain collapse operators that define the evolution of generally covariant dynamical reduction models for a real Klein-Gordon field for which the dynamical law guarantees that an initial Hadamard state, which belongs to a suitable class of states (which is a dense domain in the relevant Hilbert space), evolves into a final Hadamard state, also inside this class. This is the content of theorems \ref{MainThm} and \ref{MainThm2}. We then provide a simple example in which we compute the expectation value of the renormalized energy-momentum tensor in the Hadamard state resulting from such a collapse, and estimate the difference in the expectation values of the energy-momentum tensor in the initial and final states. We show that an appropriate choice in a parameter of the model yields a small change in the  renormalized energy-momentum. 

Our final remarks appear in section \ref{sec:Conclusion}.

The conventions of this paper are as follows: By a spacetime, $(M,g)$, we mean a real four-dimensional, connected (Hausdoff, paracompact) $C^\infty$ differentiable manifold, $M$, equipped with a Lorentzian metric $g$. We restrict our interest to those spacetimes which are time-orientable and assume a choice of time-orientation has been made and we further restrict our spacetimes to be globally hyperbolic \cite{BernalSanchez1, BernalSanchez2}. Our metric, $g$, has signature $(-,+,+,+)$.     For a subset $S\subset M$, $J^+(S)$ denotes the causal future of $S$ and $J^-(S)$ its causal past.

We use units in which $c = \hbar = k  = 1$ and use the symbol $G_\text{N}$ to denote Newton's constant. Spacetime points are denoted by Roman characters ($\x, \y, \dots$).  Complex conjugation is denoted by an overline. Concrete operators on Hilbert spaces are indicated by capital letters surmounted with carets, such as $\widehat{A}, \widehat{B}, \dots$, while elements of an abstract non-commutative algebra are indicated by caret-free letters such as $A, B, \ldots$. The adjoint of a Hilbert-space operator, $\widehat A$, is denoted by ${\widehat A}^*$.  $O(x)$ denotes a quantity for which $O(x)/x$ is bounded as $x \to 0$.

%=====================================================================================================
% DYNAMICAL REDUCTION PROBLEMS AND GENERAL COVARIANCE
%=====================================================================================================
\section{Quantum field theory and dynamical reduction models}
\label{sec:DynRed}

The idea that state reduction in quantum theory may have as a fundamental origin the quantum interactions between gravity and matter has been suggested by a number of authors -- here we mention the work of Di\'osi \cite{Diosi:1988uy, DiosiQGC} and Penrose \cite{Penrose1, Penrose1a, Penrose:1996cv, Penrose2}, as well as a proposal by one of us (the main early papers are \cite{Kay1998a, Kay1998b, KayAbyaneh}; see \cite{Kay:2015csa} for a review of recent work and further references) in which the entanglement between gravity and matter plays a fundamental r\^ole in the explanation both of state reduction and of entropy increase, {\it i.e.}, of the origin of the second law of thermodynamics.

Nevertheless, in the absence of a quantitative theory for the quantum interactions between gravity and matter, we feel it is a more realistic short-term goal to formulate relativistic dynamical reduction models in terms of a suitable version of semiclassical gravity with collapses where we expect precise mathematical questions can be formulated and answered.   This, in turn, will, need to be built on a prior understanding of QFT in CS in the presence of collapses.

In particular, we expect that we need to address new questions in QFT in CS related to the renormalizability of the energy-momentum tensor. As we mentioned in the introduction, it is well-known that, already in their non-relativistic versions, dynamical reduction models violate energy conservation, but these violations can be tuned to be small enough to be phenomenologically acceptable. Similarly, in the context of quantum field theory, dynamical reduction models will violate the conservation of energy momentum in the spacetime region where the state reduction occurs.  However, the situation seems more severe in quantum field theory, for one could imagine starting out in a Hadamard state and producing a post-collapse state that is not Hadamard. This would mean not only that energy momentum conservation is violated, but that it is violated by an infinite amount. It is the main purpose of this paper to show that this need not happen; it is possible to construct collapse models in a QFT in CS context for which the Hadamard form is indeed preserved on collapse.

The plan for the remainder of this section is as follows:  We first recall, in section \ref{QFT} in a concise way, the relevant details about the mathematical formulation of the Klein-Gordon theory in curved spacetimes.   Because this theory is linear, it is fully understood from a mathematical standpoint and thereby provides a simple and useful arena for testing dynamical reduction models in curved spacetimes. We then proceed, in section \ref{CovColl}, to recall (in section \ref{subsec:CovCollRev}) earlier work on how quantum field theory in a fixed spacetime can be modified to incorporate spontaneous collapse and then, in section \ref{subsec:choice}, we discuss  our own proposal for how this theory may be developed further in a way consistent with general covariance. 

\subsection{The Klein-Gordon theory in curved spacetimes}
\label{QFT}

On a given spacetime, $(M, g)$, the real Klein-Gordon field algebra, $\mathscr{A}$, is the $^*$-algebra with identity $\I$ generated by (smeared) fields of the form $\Phi(f)$, where $f \in C_0^\infty(M)$, satisfying the following properties: (i) $f \mapsto \Phi(f)$ is a linear map (linearity), (ii) $\Phi(f) = \Phi(f)^*$ (hermicity), (iii) $\Phi \left( \left( \Box - m^2 - \xi R \right)f \right) = 0$ (field equation) and (iv) $[\Phi(f), \Phi(g)] = -\ii E(f,g)\I$ (spacetime commutation relations), where $E = {\rm A} - {\rm R}$ is the advanced-minus-retarded fundamental solution, such that $\left( \Box - m^2 - \xi R \right) Ef = 0$.  Here, for a given test function, $f$, $\Phi(f)$ is to be interpreted as the integral over the  spacetime manifold of the (mathematically ill-defined) field at a point $\Phi(\x)$ times 
$f(\x)$ with respect to the natural volume element, $d\vol$, for the spacetime metric $g$.

We remark that, as is well known, the above spacetime commutation relations may be regarded as the result of imposing the usual canonical commutation relations on a Cauchy surface and evolving with some choice of global time function, {\it i.e.}, a globally-defined, real-valued function on the spacetime with the property that all the constant time surfaces are Cauchy; the result however being entirely independent of that choice.

A state, $\omega$, is a linear functional $\omega: \mathscr{A} \to \mathbb{C}$, which is normalized, $\omega(\I) = 1$, and positive, $\omega(A A^*) \geq 0$ for any $A \in \mathscr{A}$. 
It is fully specified by the specification of all the (smeared) \textit{n-point functions}, {\it i.e.}, of quantities of form $\omega(\Phi(f_1) \cdots \Phi(f_n))$.  Throughout this paper, we shall concentrate on a class of states which are quasi-free, which means that all the n-point functions are determined by the two-point function via the formula $\omega(\exp[\ii \Phi(f)]) = \exp[-\omega(\Phi(f) \Phi(f))/2]$.

The standard operator valued distributions and Hilbert space approach to quantum field theory can be recovered by performing the GNS construction, which produces out of the observable algebra, $\mathscr{A}$, and an algebraic state, $\omega$, a GNS triplet $(\pi, \mathscr{D} \subset \mathscr{H}, \Omega)$, where $\pi: \mathscr{A} \to \mathscr{L}(\mathscr{H})$ is a representation that maps algebra elements to linear operators acting on a dense subspace $\mathscr{D} $ of the Hilbert space $ \mathscr{H} $ and where 
$\Omega \in \mathscr{H}$ is the GNS cyclic vector 
({\it i.e.}, the set   span $\lbrace  \pi(A)\Omega, A\in   \mathscr{A}\rbrace $  is  dense  in $\mathscr{H}$). Let us call $\widehat{\Phi}(f) = \pi(\Phi(f))$. We have that
\begin{equation}
\omega(\Phi(f)\Phi(g)) = \langle \Omega | \widehat{\Phi}(f) \widehat{\Phi}(g) \Omega \rangle,
\label{KG2pt}
\end{equation}
which is the textbook expression for the Wightman two-point function (smeared in the test functions $f$ and $g$), $W(f,g)$.  Identifying  $\Omega$ with the vacuum state of the theory, we have that $\widehat{\Phi}(f) = \ii \widehat{a}({KEf}) - \ii \widehat{a}^*(KEf)$, where $(K, \mathcal{H})$ is the \textit{one-particle Hilbert space structure} associated to the quasi-free state $\omega$ (see Appendix A in \cite{Kay:1988mu}) and $\mathscr{H} = \mathscr{F}_s(\mathcal{H})$ is the symmetric Fock space over the one-particle Hilbert space $\mathcal{H}$. For $\xi\in \mathcal{H}$, $\widehat{a}(\xi)$ is the usual annihilation operator on $\mathscr{F}_s$ and $\widehat{a}^*(\xi)$ its adjoint, so that
$[\, \widehat{a}(\xi), \widehat{a}^*(\chi)]=\langle\xi|\chi\rangle\widehat{\I}$.

\subsection{Relativistic collapse}
\label{CovColl}

\subsubsection{Review of general formalism}
\label{subsec:CovCollRev}

Here we offer a brief description of, and further develop, the relativistic collapse theory which was formulated by Bedingham in \cite{Bedingham:2010hz} and generalized to the context of a fixed background curved spacetime in \cite{Bedingham:2016aus}.  

We shall assume here that our spacetime has a compact spatial section.   This will ensure that there is a globally preferred Hilbert space representation up to unitary equivalence. In fact, for such a spacetime, we know \cite{Verch, LudersRoberts} that the GNS representations of all pure quasi-free Hadamard states are unitarily equivalent and therefore we can work in a fixed one of these representations, say $\pi$ on $\mathscr H$. Also, if the spacetime has compact spatial section, the volume enclosed by an initial and a final Cauchy surface (at finite times) is finite, and a finite number of collapses occur in the state evolution from the former to the latter.\footnote{One can give arguments that our results should be extendable to spacetimes with non-compact spatial sections but this would involve us with further technicalities beyond those we are able to treat in the present paper.} We remark concerning terminology that, in Section \ref{sec:Result}, when a vector $\psi$ in $\mathscr H$ is such that the two-point function of the algebraic state, $\langle \psi | \pi ( \Phi(f) \Phi(g) ) \psi \rangle$, is Hadamard, we shall refer to $\psi$ as a \textit{Hadamard state vector}.

The basic idea is to adopt the Heisenberg picture throughout for the field operators -- {\it i.e.}\ we continue to define the field algebra as in Section \ref{QFT} -- but we assume, for a given global time function whose constant-time surfaces are Cauchy surfaces, that the state changes abruptly on certain constant-time Cauchy surfaces, $\Sigma$,\footnote{Due to these abrupt changes of the state, the Heisenberg evolution picture will {\it not} hold globally.} which contain certain randomly selected spacetime points $\x$ which we shall call \textit{spacetime collapse centers} (while the state does not change in between these abrupt changes). These spacetime collapse centers are supposed in \cite{Bedingham:2010hz} to occur according to a Poisson distribution with a fixed rate, $\mu$ (taken to be a new constant of nature) per unit spacetime volume. 

The rule for the change of state at such a Cauchy surface is taken to be of the general form of the GRW reduction rule:
\begin{align}
\psi_{\Sigma}\mapsto \psi_{\Sigma^+}  =N \widehat{L}_\x(Z) \psi_{\Sigma},
\label{RedRuleRel}
\end{align}
where $\psi_{\Sigma}$  denotes the state before the abrupt change and $\psi_{\Sigma^+}$ the state after it.  The collapse operator takes the general form: 
\begin{equation}
\widehat{L}_\x(Z) = \left(\frac{\pi}{2\alpha}\right)^{-1/4} \, \exp(-\alpha(\widehat{B}(\x) - Z\widehat{\I})^2)
\label{Lformal}
\end{equation}
and $N$ denotes the normalization factor 
\begin{equation}
N=\langle\widehat{L}_\x(Z) \psi_{\Sigma}|\widehat{L}_\x(Z)\psi_{\Sigma}\rangle^{-1/2}.
\end{equation}

Here, the collapse generator, $\widehat{B}(\x)$, is taken to be a self-adjoint operator constructed -- in a way which it will be our purpose to specify below -- out of local fields centered around the spacetime collapse-center $\x$; $\alpha$ will be called here the \textit{collapse parameter} of the theory, taken, like $\mu$, to be a fundamental constant of nature with the same dimensions as $\widehat{B}(\x)^{-2}$;  $Z$ is a real scalar constant -- which we shall sometimes call the \textit{field-space collapse center} below -- assumed to be chosen (anew at each collapse as indicated by the subscript, $\x$, on $\widehat{L}_\x$) at random with a probabilty density
\begin{equation}
\label{PROB}
dP =\langle \widehat{L}_\x(Z)\psi_{\Sigma}|\widehat{L}_\x(Z)\psi_{\Sigma}\rangle dZ.
\end{equation}
When there might be more than one spacetime collapse center, $\x$ under consideration, we will write $Z_\x$ to denote the field-space collapse center (which is randomly chosen according to the above prescription) at the spacetime collapse center $\x$.   On the other hand, we shall sometimes abbreviate the operator $L_{\x}(Z)$ by $L_{\x}$ in order to simplify our notation.
The fact that $dP$ is a probability density follows from the easily verified relation 
\begin{equation}
\int \! dZ \, \widehat{L}_\x(Z)^2 = \widehat{\I}.
\end{equation}

If the spectrum of $\widehat{B}(\x)$ is the whole real line, $\widehat{L}_\x(Z)$ will have the effect of putting the state in an approximate eigenstate of the collapse operator $\widehat{B}(\x)$ with approximate eigenvalue $Z$, thus tending to localize the state in `field space' in analogy with the way that the non-relativistic GRW model approximately localizes the non-relativistic wave function in position space.\footnote{If the spectrum of $\widehat{B}(\x)$ is only a proper subset of the real line, the localization will presumably do the best it can, but this issue seems deserving of further investigation.} We remark that it is the field-space collapse center, $Z$, here which plays an analogous mathematical r\^ole here to the center of the collapse event, $\bf z$, of nonrelativistic GRW, while what we have called here a spacetime collapse center, $\x$, plays an analogous mathematical r\^ole to one of the random times, $t_i$, at which collapses happen in nonrelativistic GRW.

Consider the state of the system at some initial Cauchy surface, $\Sigma_i$, to be given.  It can then be shown \cite{Bedingham:2016aus} that, if the following microcausality conditions hold for all space-like separated $\x$ and $\y$, 
\begin{equation}
\label{colmicro}
\hbox{(i)}\quad \! [\widehat{L}_\x(Z),\widehat{L}_\y(Z')]= 0 \quad \hbox{and (ii)}\quad \! [\widehat{L}_\x(Z),\widehat{\cal H}_{\rm int}(\y)]= 0, 
\end{equation}
equivalently,
\begin{align}
\label{MicroCaus}
\hbox{(i')}\quad \left[\widehat{B}(\x), \widehat{B}(\y)\right] =0 \quad\text{and (ii')}\quad \left[\widehat{B}(\x), \widehat{\cal H}_{\rm int}(\y)\right] =0,
\end{align}
then, given a set of spacetime collapse centers $\{x_j|\Sigma_f \succ {x_j}\succ \Sigma_i\}$ (with labels $j= 1,\ldots,n$, corresponding to an arbitrary total ordering, which respects the causal ordering of the spacetime) occurring between `initial' and `final' Cauchy surfaces $\Sigma_i$ and $\Sigma_f$, 
and given the set of all field-space collapse centers, $Z_{\x_j}$, chosen at these spacetime collapse centers, $\{Z_{\x_j}|\Sigma_f \succ {\x_j}\succ \Sigma_i\}$, the state dynamics leads to an unambiguous change of state between $\Sigma_i$ and $\Sigma_f$. Moreover, the probability rule that specifies the joint probabilities  of full sets of spacetime collapse centers $\{Z_{\x_j}|\Sigma_f \succ \x_j\succ \Sigma_i\}$  does not depend on the choice of global time function (assumed to have $\Sigma_i$ and $\Sigma_f$  as two of its constant-time surfaces) for which the intermediate Cauchy surfaces, on which the collapses happen, are constant-time surfaces provided  that the law  characterizing  the   probability distribution of the spacetime locations itself does not select a preferred time-function.  A  Poisson distribution is one possibility for such a rule.

Therefore, provided the equation \eqref{colmicro}, or equivalently \eqref{MicroCaus}, holds and  given that  our dynamical and (statistically formulated) collapse rules make no reference to any preferred global time function, the whole set of (overall statistically formulated) dynamical rules will be  covariant.

\subsubsection{Our proposal for the choice of collapse generators.}
\label{subsec:choice}
  
In this paper, we take the point of view that the choice of the collapse operator should be made according to the following guiding principles:

\begin{enumerate}
\item Localization: In quantum field theory, processes are not sharply localized at space-time points, but rather smeared in compact spacetime regions.

\item Causality: In bosonic quantum field theory, any collapse operator must respect Einstein causality through the CCR condition of the field algebra.

\item Covariance: Any collapse operator must be constructed in a general covariant way and avoid Lorentz violations.

\end{enumerate}

We remark that, in the non-relativistic theory, locality is less of a problem thanks to the existence of a position operator.   The guiding principles above will compensate for the lack of such a position operator in a relativistic context.  In any case, these principles for relativistic QFT in CS seem to allow one to construct a large class of possible models as follows.  First, if we have a Poisson distribution that selects spacetime collapse centers $\{\x_i\}_I$, we let $f_{\x_i} \in C_0^\infty(M)$ be smooth functions of compact support peaked around each one of the spacetime collapse centers ${\rm x}_i \in M$.\footnote{An obvious issue that arises here is that we need to define what we mean by the `same' $f_{\x_i}$ around two different spacetime collapse centers, $\x_i$.   We wish to simply remark here that there are various ways in which this can be done and do not enter into details. For an example, see \cite{Bedingham:2016aus}.}

Then by choosing collapse generators of the form $\Phi(f_{\x_i})$,
one can produce a model that satisfies all of our principles. More generally, one can choose polynomials of field operators that are covariantly smeared in the smooth functions of compact support $f_{\x_i}$, in accordance with the following definition:

\begin{defn}
We call a collapse generator a \textit{covariantly smeared polynomial collapse generator in} $f_\x$ if it is a polynomial in fields smeared against test functions covariantly constructed out of the metric, its inverse and their derivatives, the test function $f_\x$ and its covariant derivatives.
\label{def:covpol}
\end{defn}

We shall refer to such collapse generators as \textit{covariant polynomial collapse generators}. 

Some examples of covariant polynomial collapse generators are $\alpha \Phi(f{_\x})$, $\delta  \Phi( R f{_\x})$ and $\beta \Phi(\Box f{_\x}) \Phi(f{_\x})  + \gamma \Phi(f{_\x})$. By the linearity of the field, equalities such as  $\Phi(\alpha R f{_\x} + \beta \Box  f{_\x})\Phi(f{_\x}) = \alpha \Phi( R f{_\x}) \Phi(f{_\x}) + \beta \Phi(\nabla^c\nabla_c  f{_\x})\Phi(f{_\x})$ hold. We define derivatives of the field weakly, {\it e.g.}, $(\nabla_a \Phi)(g^{ab} \nabla_b f{_\x}) = -\Phi(\Box f{_\x})$, so, {\it e.g.}, by the field equation, $\Phi\left(\left( \Box - m^2 - \xi R \right)f{_\x}\right) = 0$.

The reduction rule, eq. \eqref{RedRuleRel}, generalizes as follows:  Let $\alpha$ be the collapse parameter of the theory.   Given a set of spacetime collapse centers, $\{ \x_i \}_I$, that have been chosen according to the Poisson distribution discussed above, let $Z \in \mathbb{R}$
be a real-scalar constant, randomly chosen (according to \eqref{PROB})  for each spacetime collapse center, and centering around each one of these spacetime collapse centers and choosing a fixed smooth function of compact support, $f_{\x_i}$, define the evolution law at a constant-time Cauchy surface for a choice of global time function passing through any one of these spacetime collapse centers by
\begin{subequations}
\begin{align}
\psi  \mapsto \psi_{\x_i} & = \frac{\widehat{L}_{\x_i}(Z) \psi }{\langle \widehat{L}_{\x_i}(Z) \psi | \widehat{L}_{\x_i}(Z) \psi \rangle^{1/2}}, \, \, \text{ with }  \\
\, \, \widehat{L}_{\x_i}(Z) & = \left(\frac{\pi}{2\alpha}\right)^{-1/4}\exp\left[-\alpha\Big(\widehat{Q}(\widehat{\Phi}, f_{\x_i}) - Z\widehat{\I} \Big)^2 \right],
\end{align}
\label{psix}
\end{subequations}
\\ where $Q(\Phi, f_{\x_i})$ is a fixed covariantly smeared polynomial collapse generator in $f_{\x_i}$. 
$Z$ will again be called the \textit{field-space collapse center}.    

It may seem somewhat strange (and partly teleological) that the collapse operator which governs the collapses at a particular Cauchy surface apparently depends on the values of the quantum field in some spacetime region, {\it i.e.} the union of the supports of the smearing functions involved in its definition, which includes parts of both the future and the past of that surface.   However, one can of course take the point of view that the observables are localized {\it on} the Cauchy surface.   For example, using the terminology of  \cite{Kay:1988mu}, the `covariantly smeared' field $\Phi(f)$ for some smooth compactly supported test function, $f$, is equal to the `symplectically smeared' field 
$\sigma(\Phi, \phi_c)$ where $\phi_c$ is the classical solution $Ef$.   In other words it is $\varphi(p) - \pi(f)$ where $(\varphi, \pi)$ are the Cauchy data of the quantum field, $\Phi$, and $(f,p)$ the Cauchy data of the classical solution $\phi_c=Ef$ on the Cauchy surface.    

A word of caution is due:  Even when two spacetime collapse centers, $\x_1, \x_2$ are spacelike separated, we would expect there to be points in the support of $f_{\x_1}$ and points in the support of $f_{\x_2}$ which are not spacelike separated.  Thus microcausality of our Klein-Gordon field algebra is not sufficient for commutativity of $L_{\x_1}$ and $L_{\x_2}$.  So, in view of eq.\ \eqref{colmicro} or, equivalently, \eqref{MicroCaus}, and the discussion around these equations, we have no guarantee that the dynamical evolution will be independent of the global time-function.

In our opinion, there are three ways out of this problem. The first solution has been provided by Bedingham \cite{Bedingham:2010hz} and introduces additional non-dynamical fields. As for the other two, we refrain from providing a rigorous formulation of these alternatives, but rather explain the lines of thought in general terms. 

The first alternative consists of dealing with smooth functions with diamond-like support for the field smearings and providing a unique prescription for the ordering of the collapse generators, whereby if ${\rm supp}(f_{\x_i})$ and ${\rm supp}(f_{\x_j})$ are not spacelike separated  and ${\rm supp}(f_{\x_i}) \cap J^+({\rm supp}(f_{\x_j})) = \emptyset$, then the collapse operators are ordered as $T_c(L_{\x_i} L_{\x_j}) = L_{\x_j} L_{\x_i}$, where by $T_c$ we mean collapse time-ordering, {\it i.e.}, the collapse operator $L_{\x_i}$ acts before $L_{\x_j}$. The second alternative can be loosely stated as requiring that the supports of the functions $f_{\x_i}$ be sufficiently small and that the collapse events be sufficiently scarce, such that ordering problems do not occur.\footnote{This would mean heuristically that support of $f_\x$ is so small that state reductions occur almost `at a single point'.} This could be achieved by modifying the Poisson distribution yielding the spacetime collapse centers $\x_i$, but for the moment we refrain from being more precise than this, as this is not the purpose of this paper.

%=====================================================================================================
% THE HADAMARD CONDITION
%=====================================================================================================
\section{The Hadamard condition and the energy-momentum renormalization}
\label{sec:Hadamard}

In this section, we review the Hadamard condition on the two-point function of states of the Klein-Gordon field, as was first formulated by one of us and Wald in 
\cite{Kay:1988mu} in a rigorous fashion.\footnote{We remind the reader that $(M,g)$ is assumed to be globally hyperbolic and in this paper we deal only with this case. For the definition of a Hadamard state in asymptotically AdS spacetimes, see the recent paper \cite{Dappiaggi:2017wvj}.} This serves several purposes: First, it allows for this work to be as self-contained as possible. Second, it introduces the distributional language of quantum field theory. Third, it allows one to prove the smoothness of certain quantities that will appear in our main result (see appendix \ref{app:A}). We then review how the renormalized energy-momentum tensor is defined by a point splitting prescription and how its expectation value in any Hadamard state satisfies Wald's axioms \cite{Wald:1995yp}.

\subsection{The Hadamard condition}

The Hadamard condition is a property of a state of a given linear theory which allows one to obtain certain renormalized non-linear observables, such as the energy-momentum tensor, which do not belong to the minimal algebra ({\it cf.}\ chapter 3 in \cite{Kay:1988mu}) of (essentially) sums of products of smeared fields. In this subsection, we return to our discussion of the Klein-Gordon theory and we shall state precisely what the Hadamard condition is for a state on this theory.

Two-point functions such as defined in equation \eqref{KG2pt}, typically arise from unsmeared two-point functions using the standard $\ii \epsilon$ ``integrate then take the limit" prescription,
\begin{align}
\omega_2(\Phi(f)\Phi(g)) & = \lim_{\epsilon \to 0^+} \!\! \int_{M \times M} \!\!\!\!\!\!\!\!\!\!\!\!\! d\vol(\x) \, d\vol(\y) \, f(\x) g(\y) W_2^\epsilon(\x,\y),
\label{omega2}
\end{align}
where $W_2^\epsilon(\x, \y)$ is a two-point function with a suitable small imaginary part. For example, in the familiar case of a massless scalar field in four-dimensional Minkowski spacetime in the Minkowski ground state, $W_2^\epsilon(\x, \y)$ is\footnote{More generally, the $\ii \epsilon \left( t-t' \right)$ term can be replaced by any $\ii \epsilon \left( T(\x)-T(\x') \right)$, where $T$ is an arbitrary future increasing time function in the spacetime.} 
\begin{align}
W^{{\rm M}}_\epsilon \! \left((t, \mathbf{x}); (t', \mathbf{x'})\right) \! \! = \! \frac{1/(4\pi)}{-\!\left( t-t' \right)^2 \!\! + \! \left| \mathbf{x}-\mathbf{x}' \right|^2 \!\! + \! \ii \epsilon \left( t-t' \right)\! + \epsilon^2 }. 
\end{align}

The integrand on the right-hand side of eq. \eqref{omega2} is integrable for each $\epsilon > 0$, so the left-hand side exists if the limit exists. Moreover, the anti-symmetric part of the two-point distribution is fixed by the CCR, $\omega_2(f,g) - \omega_2(g,f) = -\ii E(f,g)$, while  the symmetric part  is  fully determined by the state. It is the ultraviolet behaviour of the symmetric part of the two-point distribution that provides the criterion as to whether a state is Hadamard. In order to provide the definition for a Hadamard state, we first state two useful geometric definitions:

\begin{defn}
A {\it convex normal neighborhood} $U \subset M$ is an open set such that for any pair of points $\x, \y \in U$, there exists a unique geodesic from $\x$ to $\y$ fully contained in $U$.
\label{def:ConvexNN}
\end{defn}

In convex normal neighborhoods, one can define unambiguously the squared geodesic distance between two points in $(\x, \y) \in U \times U$, which, following \cite{Kay:1988mu} we denote by $\sigma(\x,\y)$.

\begin{defn}
Let $C \subset M$ be a Cauchy surface of $M$. We say that  the open $N \subset M$ is a {\it causal normal neighborhood of} $C$ if $C \subset N$ and for any pair of points $\x, \y \in N$ such that $\x \in J^+(\y)$, there exists a convex normal neighborhood containing $J^-(\x) \cap J^+(\y)$.
\label{def:CausalNN}
\end{defn}

Also, following \cite{Kay:1988mu}, we note that for any Cauchy surface $C \subset M$ it is always possible to find a causal normal neighborhood $N$ such that $C \subset N$.

We are now ready to state the definition \cite{Kay:1988mu} of a Hadamard state.

\begin{defn}[{\bf Hadamard state}]

Let  $(M,g)$  be  a spacetime and let $T$ be a global time function, increasing to the future, whose constant time surfaces are Cauchy surfaces. Let $N$ be a causal normal neighborhood of some such Cauchy surface, $C$, and let $O$ be an open neighborhood of $N\times N$.  Further, let $O' \subset N \times N$ be an open neighborhood whose closure is contained in $O$. Now, let $\chi \in C^\infty(N \times N)$ be an interpolating function, such that $\chi(\x, \y) = 1$ if $(\x,\y) \in O'$ and $\chi(\x, \y) = 0$ if $(\x,\y) \notin O$. We say that the state $\omega:\mathscr{A} \to \mathbb{C}$ is a {\it Hadamard state} if, for each $n \in \mathbb{N}$ and $\epsilon > 0$, there  exists  a bi-function   $ W^n \in C^n(M \times M)$  such  that,
\begin{widetext}
\begin{align}
\omega_2(f,g) & = \lim_{\epsilon \to 0^+} \int_{M \times M} \! d\vol(\x) \, d\vol(\y) \, f(\x) g(\y) \, \left[\chi(\x, \y) H_\epsilon^{T, n}(\x,\y) + W^n(\x,\y) \right],
\label{Had1}
\end{align}
\end{widetext}
with  $H_\epsilon^{T,n}: N \times N \to \mathbb{C}$ defined as
\begin{align}
H_\epsilon^{T, n} (\x,\y) \! = \! \frac{1}{(2 \pi)^2} \! \left(\frac{\Delta^{1/2}(\x, \y)}{\sigma_\epsilon(\x, \y)} \! + \! V^{(n)}(\x, \y) \ln[\sigma_\epsilon(\x, \y) ] \right).
\label{Parametrix}
\end{align}
Here,  the logarithm branch cut is taken along the negative real axis, $\sigma_\epsilon(\x, \y) = \sigma(\x, \y) + 2 \ii \epsilon[T(\x) - T(\y)] + \epsilon^2$, $\Delta$ is the van Vleck-Morette determinant \cite{deWittBrehme} and where each $V^{(n)}= \sum_{k=0}^n v_k \sigma^k$ is a smooth bi-function with smooth bi-function coefficients $v_k$ determined by the Hadamard recursion relations (see \cite{deWittBrehme, Gara},  see also \cite{Decanini:2005gt}) up to order $n$, which guarantees that $H_\epsilon^{T, n}$ is a Green function of the Klein-Gordon equation to order $n$.
\label{Def:Hadamard}
\end{defn}

Several comments are in place: First, because $W^n$ is $C^n$ for each $n$, with $n$ as large as desired, this contribution can be taken to be smooth \cite{Kay:1988mu}. Second, the definition is independent of the time function, $T$, the chosen Cauchy surface, $C$,  the interpolating function $\chi $ and the chosen  causal normal neighborhood of $C$, $N$. See \cite[sec. 3.3]{Kay:1988mu}.

In light of eq. \eqref{omega2}, the Hadamard condition can be formally stated as
\begin{equation}
\omega^\epsilon(\Phi(\x)\Phi(\y))= \omega_2^\epsilon(\x, \y) = \chi(\x, \y) H_\epsilon^{T, n}(\x,\y) + W^n(\x,\y)
\label{FormalHad}
\end{equation}
for each $n$, but we would like to stress that definition \ref{Def:Hadamard} takes into account all the distributional and geometric subtleties that give sense to the formal equation \eqref{FormalHad}.

\subsection{The renormalization of the energy-momentum tensor}

That the energy-momentum tensor is not in the Klein-Gordon field algebra can be seen immediately already because it is involves the product of fields at the same point before smearing. The construction of such observables out of free fields requires a renormalization prescription.

A covariant renormalization axiomatic prescription for the energy-momentum tensor has been given by Wald in what is now known as the Wald axioms \cite[sec. 4.6]{Wald:1995yp}, {\cite{Wald78}: (i) If $\omega_1(\Phi(\x)\Phi(\y)) - \omega_2(\Phi(\x)\Phi(\y))$ is a smooth function, then $\omega_1\left(T^{{\rm ren}}_{ab}(\x)\right) - \omega_2\left(T^{{\rm ren}}_{ab}(\x)\right)$ is smooth by a splitting prescription, (ii) $T^{{\rm ren}}_{ab}$ is local with respect to the state of the field and invariant under globally hyperbolic isometric embeddings, (iii) For all states, $\nabla^a \omega\left(T^{{\rm ren}}_{ab}\right)) = 0$ and (iv) $\omega_{{\rm M}}\left(T^{{\rm ren}}_{ab}\right) = 0$ in the Minkowski vacuum. 

A renormalization point-splitting scheme that satisfies Wald's axioms consists on constructing the {\it Hadamard parametrix}\footnote{Whenever the spacetime is not analytic, convergence of the asymptotic series can be ensured with the aid of a set of cut-off functions. See \cite{Hack:2012qf} and references therein.} from the Hadamard recursion relations,

\begin{align}
& H(\x,\y) = H_0^{T,\infty}(\x,\y) \nonumber \\
&= \! \lim_{\epsilon \to 0} \! \frac{1}{(2 \pi)^2} \! \left(\!\!\! \frac{\Delta^{1/2}(\x, \y)}{\sigma_\epsilon(\x, \y)} \! + \!\! \sum_{k=0}^\infty \! v_k(\x, \y) \sigma^k(\x, \y) \ln[\sigma_\epsilon(\x, \y) ] \!\! \right)\!\!  ,
\end{align}
and obtaining the renormalized energy-momentum tensor by a point-splitting restriction with respect to the Hadamard parametrix as follows.   Let $\omega$ be a Hadamard state.  Then 
\begin{widetext}
\begin{equation}
\omega(T^{{\rm ren}}_{ab}(\x))  = \lim_{\x' \to \x} \left\{ \left[ \nabla_a \nabla_{b'} -\frac{1}{2} g_{ab}(\x)\left(\nabla_c\nabla^{c'} + m^2 + \xi R(\x) \right) \right] G\left(\x, \x'\right) + P_{ab}(\x) \right\},
\label{Tren}
\end{equation}
\end{widetext}
where 
\begin{equation}
G(\x, \x')  = \omega(\Phi(\x)\Phi(\y)) - H(\x,\y).
\end{equation}

Here, $P_{ab}$ is a certain local, symmetric tensor correction term ($(1/32\pi^2)[v_1]_cg_{ab}$ in the notation of \cite{Wald78}) introduced in \cite{Wald78}, (see also \cite{KhavkineMoretti}) needed in order for the covariant conservation equation $\nabla_a\omega(T_{ab}^{\rm ren}$ to hold and giving rise, in the approach of \cite{Wald78}, to the trace-anomaly (see \cite{Wald78, Birrell:1982ix}) in the case the of conformal coupling.   

As explained in \cite{Wald78, Birrell:1982ix} and in the sense explained in those references, $T^{{\rm ren}}_{ab}(\x)$ is understood to be ambiguous up to the addition of arbitrary linear combinations of `local geometrical terms', $g_{ab}$, $G_{ab}$, $^1H_{ab}$ and $^2H_{ab}$, which come from adding arbitrary linear combinations of $\sqrt g$, $\sqrt g R$ and two further terms involving higher derivatives of the metric to the Lagrangian.  
By construction, $T^{{\rm ren}}$ satisfies Wald's axioms: The Hadamard parametrix guarantees that (i) and (ii) hold, while, as we discussed above, the correction term $P_{ab}$  ensures (iii). (iv) can always be ensured by taking advantage of the, just mentioned, freedom to add multiples of $g_{ab}$.

The defining equations \eqref{Tren} provide a prescription for making sense of the semiclassical Einstein field equations,
\begin{equation}
G_{ab} = 8 \pi G_\text{N} \, \omega\left(T^{{\rm ren}}_{ab}\right),
\end{equation}
for computing back-reaction effects. To be precise, as discussed above, one needs to consider higher order derivative terms in the metric at this stage \cite{Birrell:1982ix}.

%=====================================================================================================
% THE MAIN RESULTS
%=====================================================================================================
\section{The main results}
\label{sec:Result}

In section \ref{sec:DynRed} we argued that any collapse operator of the dynamical reduction model must be constructed covariantly. We proposed that this preferred operator be generated by a polynomial in the field smeared against a geometric and covariant polynomial constructed out of the metric and its derivatives acting on the test function $f_\x$. We referred to this class of collapse generators as covariant polynomial. (See definition \ref{def:covpol} in section \ref{subsec:choice} for the precise definition.) 

The purpose of this section is to state and prove the main results of this paper, which show that the states resulting from the collapse of a large class of Hadamard states are themselves Hadamard states, when the collapse generator is such a covariant polynomial.

More precisely, in the first part of this section, in \ref{subsec:Theorem}, we show that given a covariantly smeared \textit{monomial} collapse generator, {\it i.e.} a collapse generator that is linear in the field, and given an initial Hadamard state vector in our Hilbert space\footnote{\label{single} See section \ref{subsec:CovCollRev} for the definition of the term `Hadamard state vector' and recall that, as explained there, we we will be able to work in a fixed Hilbert space representation in which all pure quasi-free Hadamard states arise from vector states thanks to our assumption that our spacetime has a compact spatial section.}, which belongs to a certain dense domain (the definition of which, in turn, depends on a fixed but arbitrary choice of another quasi-free Hadamard state vector which we call $\Omega$ below) the post-reduction state vector will also be Hadamard. We further show that the post-reduction state vector remains Hadamard on a certain enlarged domain which is an invariant domain for the collapse operator and therefore any number of successive collapses will also result in Hadamard states. The collection of these observations is our first main result, theorem \ref{MainThm}.

In the second part of this section, subsection \ref{subsec:Pert}, we weaken the linearity assumption and prove with, however, a slightly lower standard of rigor that a perturbative version of the Hadamard condition holds for any general covariant polynomial operators of finite order. This is our second main result, theorem \ref{MainThm2}.

We conclude this section with an example that connects our results with the discussion of section \ref{sec:Hadamard} on the renormalizability of the energy-momentum tensor. This is the content of subsection \ref{subsec:Ex}.

\subsection{A result for linear covariant polynomial collapse generators}
\label{subsec:Theorem}

\begin{thm}
Let $(\pi, \cal{D} \subset \mathscr{H},$ $\Omega)$ be the GNS triple of the Klein-Gordon theory with respect to a quasi-free algebraic Hadamard state $\omega: \mathscr{A} \to \mathbb{C}$ on the real Klein-Gordon field algebra for our spacetime as defined in section \ref{QFT}, and let $\mathscr{D}$ be the set of vectors, $\psi$ which arise as finite sums of form
\begin{equation}
\label{lincomb}
\psi = \sum_{n=1}^N c_n e^{i\widehat{\Phi}(g_n)^c}
\end{equation}
for arbitrary $N$ and arbitrary functions $g_i$, $i=1 \dots N$ in $C_0^\infty(M)$\footnote{\label{dense} In other words, $\mathscr{D}$ consists of finite linear combinations of  Weyl operators acting on the `vacuum' vector, $\Omega$, or, in yet other words, to finite linear combinations of coherent states built on $\Omega$.   In view of well-known properties of Weyl operators -- equivalently in this context of coherent states -- $\mathscr{D}$ is, therefore, itself a dense domain.} where (here and throughout) we take $\widehat{\Phi}(g)$, for any $g\in C_0^\infty(M)$ to denote 
$\pi(\Phi(g))$ on the domain $\cal D$ and $\widehat{\Phi}(g_n)^c$ to denote its closure.   The latter will be self-adjoint by the fact that -- see, {\it e.g.}, section 5.2.4 in \cite{KhavkineMoretti} -- for all $g\in C_0^\infty(M)$, $\widehat\Phi(g))$ is essentially self-adjoint on the domain $\cal D$. Let $f$ be a particular choice of $C_0^\infty$ test function on $M$ and take $\widehat{L}_\z: \mathscr{D} \to \mathscr{H}$ to be the self-adjoint, bounded (`collapse') operator, defined by 
\begin{equation}
\widehat{L}_\z = \exp\left[-\alpha\Big(\widehat{\Phi}(F)^c - Z \widehat{\I} \Big)^2 \right]
\label{Lz}
\end{equation}
with $\alpha >0$, $Z \in \mathbb{R}$. We have that:
\begin{enumerate}
\item Any normalized $\psi  \in \mathscr{D}$ is a Hadamard state vector. It follows that the vector $\psi_\z  = \widehat{L}_\z\psi /\langle \widehat{L}_\z \psi  \, | \, \widehat{L}_\z \psi \rangle^{1/2}$ is a Hadamard state vector.
\item Let $\mathscr{G}$ be the dense subset of $\mathscr{H}$ consisting of finite linear combinations of vectors of the form 
$\psi_n = \prod_{k=0}^{n} \widehat{L}_{\z_k} \psi$ where we adopt the convention (here and throughout) that the right hand side means $\widehat{L}_{\z_0}\widehat{L}_{\z_1}\dots\widehat{L}_{\z_n} \psi$, where $\psi \in \mathscr{D}$ and let
\begin{equation}
\widehat{L}_{\z_k} = \exp\left[-\alpha\Big(\widehat{\Phi}(F^k)^c - Z_k \widehat{\I} \Big)^2 \right],
\label{Lzk}
\end{equation}
be collapse operators labelled by the non-negative integer $k$, then, if any normalized $\psi \in \mathscr{G}$ is a Hadamard state vector, $\psi_\z = \widehat{L}_\z\psi /\langle \widehat{L}_\z \psi  \, | \, \widehat{L}_\z \psi \rangle^{1/2}$ will also be a Hadamard state vector.
\end{enumerate}
\label{MainThm}
\end{thm}

We remark that in the application to the relativistic collapse scheme of section \ref{subsec:choice}, the test function $F$ will arise in the form $P_g f_\z$ and the test functions, $f^k$ of item {\it 2} above will arise in the form $P_g^kf_{\z_k}$.  Note also, regarding item {\it 2} that $\mathscr{D} \subset \mathscr{G} \subset \mathscr{H}$, where $P_g$ and $P_g^k$ are operators covariantly constructed out of the metric, its inverse and their derivatives, such that the collapse generator is covariantly smeared in $f_\z$ and $f_\z^k$ respectively.

Item {\it 1} in the above theorem means that a single collapse of a Hadamard state, which belongs to the dense domain, $\mathscr{D}$ of the Hilbert space of the theory, produces a Hadamard state vector. Item {\it 2} guarantees that successive collapses also yield a Hadamard state vector. We now prove theorem \ref{MainThm}.

We remark that the domains, $\mathscr{D}$ and $\mathscr{G}$, each depend on the choice of quasi-free Hadamard algebraic state, $\omega$ and therefore given that we can and do (see footnote \ref{single})) regard each of their cyclic state vectors, $\Omega$,  as belonging to the Hilbert space of our chosen representation, we have many domains, say, $\mathscr{D_\Omega}$ and $\mathscr{G_\Omega}$ for each such $\Omega$, each of which (see footnote \ref{dense}) is dense by itself and our theorem therefore guarantees that any initial Hadamard state vector in the \textit{union} of all the $\mathscr{D_\Omega}$ or, indeed, in the, larger, \textit{union} of all the $\mathscr{G_\Omega}$ will be mapped by any of our collapse operators, $\widehat{L}_\z$ for any $\z$, into a Hadamard state vector (where, moreover, we know that when $\psi$ belongs to $\mathscr{G_\Omega}$ for a particular $\Omega$ then $\widehat{L}_\z\psi$ will belong to $\mathscr{G_\Omega}$ for the same $\Omega$.   However, it remains an open question whether every state vector in our Hilbert space is mapped into a Hadamard state by each or any $\widehat{L}_\z$.

\begin{proof}

We recall from section \ref{QFT} that $\pi$ will act so that $\pi(\Phi(f)) (= \widehat{\Phi}(f)) = (\ii\, \widehat{a}(KEf) - \ii\, \widehat{a}^*(KEf))$.  We denote the positive polarization of the field in this representation by $\widehat{\Phi}^+(f) = - \ii\, \widehat{a}^*(KEf)$ and the negative polarization by  $\widehat{\Phi}^-(f) = \ii\, \widehat{a}(KEf)$. We remark here that, while all these operators are originally only defined on $\cal D$, they also have an obvious meaning as operators on $\mathscr D$.

That the operator $\widehat{L}_\z: \mathscr{D} \to \mathscr{H}$ defined by eq. \eqref{Lz} is self-adjoint follows from the self-adjointness of $\widehat{\Phi}(F)^c$. That it is bounded follows from the fact that $||\widehat{L}_\z \psi|| = \langle \widehat{L}_\z \psi | \widehat{L}_\z  \psi \rangle^{1/2} \leq ||\psi||$.

\vspace*{0.2cm}
\textbf{Proof of item } {\it 1}.
\vspace*{0.2cm}

First, we show that $\psi  \in \mathscr{D}$ given by 
\begin{equation}
|\psi \rangle = \frac{\sum_{n = 0}^N c_n | \ee^{\ii \widehat{\Phi}(f_n)} \Omega \rangle}{\left( \sum_{i = 0}^N \sum_{j = 0}^N \overline{c_i} c_j \langle \ee^{-\ii \widehat{\Phi}{f_i}} \Omega | \ee^{\ii \widehat{\Phi}(f_j)} \Omega \rangle\right)^{1/2}}
\end{equation}
is a Hadamard state, {\it i.e.}, that $\omega_\psi(\x,\y) = \langle \psi | \widehat{\Phi}(\x) \widehat{\Phi}(\y) \psi \rangle$ has Hadamard form. To this end, we write
\begin{equation}
\omega_\psi(\x, \y) = \langle \psi | \left( :\widehat{\Phi}(\x) \widehat{\Phi}(\y): + \left[ \widehat{\Phi}^-(\x), \widehat{\Phi}^+(\y) \right] \right) \psi \rangle,
\end{equation}
where $: \cdot :$ denotes normal ordering with respect to $|\Omega \rangle$. Because $\left[ \widehat{\Phi}^-(\x), \widehat{\Phi}^+(\y) \right]$ is a c-bidistribution times the identity operator, we have that
\begin{align}
\omega_\psi(\x, \y) & = \langle \psi |  :\widehat{\Phi}(\x) \widehat{\Phi}(\y): \psi \rangle + \langle \Omega | \left[ \widehat{\Phi}^-(\x), \widehat{\Phi}^+(\y) \right] \Omega \rangle  \nonumber \\
& = \langle \psi |  :\widehat{\Phi}(\x) \widehat{\Phi}(\y): \psi \rangle + \langle \Omega |  \widehat{\Phi}(\x) \widehat{\Phi}(\y) \Omega \rangle,
\label{had1}
\end{align}
and the second term on the right-hand side of eq. \eqref{had1} is of Hadamard form because the algebraic state, $\omega$, is Hadamard. We are left to show that the first term on the right-hand side of eq. \eqref{had1}, given by
\begin{align}
\langle \psi |  :\widehat{\Phi}(\x) \widehat{\Phi}(\y): \psi \rangle & = \langle \psi |  \left( \widehat{\Phi}^+(\x) \widehat{\Phi}^+(\y) + \widehat{\Phi}^+(\x) \widehat{\Phi}^-(\y)  \right. \nonumber \\
& \left. + \widehat{\Phi}^+(\y) \widehat{\Phi}^-(\x) + \widehat{\Phi}^-(\x) \widehat{\Phi}^-(\y) \right)  \psi \rangle,
\label{Had1.2}
\end{align}
is smooth. This can be shown using the commutator relation $\left[\widehat{\Phi}^\pm(\x), \ee^{\ii \widehat{\Phi}(f)} \right] = \ii \left[\widehat{\Phi}^\pm(\x), \widehat{\Phi}^\mp(f)\right] \ee^{\ii \widehat{\Phi}(f)^c}$, and noticing, using lemma \ref{LemmaComm} in appendix \ref{app:A1}, that this commutator is of the form of a smooth function multiplying the Weyl operator, $\ee^{\ii \widehat{\Phi}(f)^c}$. We demonstrate how to handle the second term in \eqref{Had1.2}. The rest of the terms can be handled similarly.
\begin{widetext}
\begin{align}
 \langle \psi |  \widehat{\Phi}^+(\x) \widehat{\Phi}^-(\y) \psi \rangle & = \frac{\sum_{n = 0}^N \sum_{m = 0}^N \overline{c_n} c_m\langle \widehat{\Phi}^-(\x) \ee^{\ii \widehat{\Phi}(f_n)} \Omega |   \widehat{\Phi}^-(\y) \ee^{\ii \widehat{\Phi}(f_m)} \Omega \rangle}{ \sum_{i = 0}^N \sum_{j = 0}^N \overline{c_i} c_j \langle \ee^{\ii \widehat{\Phi}(f_i)} \Omega | \ee^{\ii \widehat{\Phi}(f_j)} \Omega \rangle} \nonumber \\
& = \frac{\sum_{n = 0}^N \sum_{m = 0}^N \overline{c_n} c_m\langle  \Omega | \left[\ee^{-\ii \widehat{\Phi}(f_n)},  \widehat{\Phi}^+(\x)\right]  \left[\widehat{\Phi}^-(\y), \ee^{\ii \widehat{\Phi}(f_m)} \right] \Omega \rangle}{ \sum_{i = 0}^N \sum_{j = 0}^N \overline{c_i} c_j \langle \ee^{\ii \widehat{\Phi}(f_i)} \Omega | \ee^{\ii \widehat{\Phi}(f_j)} \Omega \rangle}.
\end{align}
\end{widetext}

We now show that $\widehat{L}_\z \psi$ is also a Hadamard state vector. This result follows from a lemma that we now state and prove.

\begin{lemma}
Let $\pi, \mathscr{D} \subset \mathscr{H}, \Omega$ be as previously defined. Let $\widehat{L}: \mathscr{D} \to \mathscr{H}$ be a self-adjoint operator on the Hilbert space, such that for any $f, g \in C_0^\infty(M)$, 
\begin{enumerate}[(i)]
\item the commutator 
\begin{equation}
[\widehat{L}, \widehat{\Phi}^\pm(f)] = \int_M d\vol(\x) \,f(\x) \, [\widehat{L}, \widehat{\Phi}^\pm(\x)]
\end{equation}
defines an operator on the Hilbert space times a $C^\infty(M)$ function (namely $[\widehat{L}, \widehat{\Phi}^\pm(\x)]$), and
\item the double commutator 
\begin{align}
[[\widehat{L}, \widehat{\Phi}^\pm(f)], \widehat{\Phi}^\pm(g)] \nonumber \\ = \int_{M\times M} \! d\vol(\x) d\vol(\y) \,f(\x) \, g(\y) [[\widehat{L}, \widehat{\Phi}^\pm(\x)], \widehat{\Phi}^\pm(\y)],
\end{align} 
defines an operator on the Hilbert space times a $C^\infty(M \times M)$ bi-function, $[[\widehat{L}, \widehat{\Phi}^\pm(\x)], \widehat{\Phi}^\pm(\y)]$.
\end{enumerate}

Then, if $\psi  \in \mathscr{D}$ is a Hadamard state vector, it follows that $\psi_c  = \widehat{L} \psi /\langle \widehat{L} \psi \widehat{L} \psi \rangle^{1/2}$ is a Hadamard state vector.
\label{Lem:Main1}
\end{lemma}

\begin{proof}[Proof of lemma]
We denote the two-point distribution in the state $\psi_c$ by $\omega_2^c(f,g) =\langle\psi_c |\widehat{\Phi}(f) \widehat{\Phi}(g) \psi_c\rangle$. To verify whether this expression has Hadamard form, it suffices to replace the product $\widehat{\Phi}(f) \widehat{\Phi}(g)$ by its normal ordered counterpart with respect to the initial state $\psi$, and seek to verify that the expression
\begin{equation}
\langle \psi_c | : \widehat{\Phi}(\x) \widehat{\Phi}(\y) : \psi_c \rangle \in C^\infty(M \times M).
\label{NormOrd}
\end{equation}

If \eqref{NormOrd} holds, then the state vector $\psi_c$ is Hadamard. This follows from the normal ordering prescription $
\widehat{\Phi}(\x) \widehat{\Phi}(\y) = : \widehat{\Phi}(\x) \widehat{\Phi}(\y) : +   \left[ \widehat{\Phi}^-(\x), \widehat{\Phi}^+(\y) \right]$ and the relation
\begin{align}
\left[ \widehat{\Phi}^-(\x), \widehat{\Phi}^+(\y) \right] & = \langle \psi | \left[ \widehat{\Phi}^-(\x), \widehat{\Phi}^+(\y) \right] \psi \rangle \widehat{\I} \nonumber \\
& = \langle \Omega | \widehat{\Phi}(\x) \widehat{\Phi}(\y) \Omega \rangle \widehat{\I}
\label{NormOrdHad}
\end{align}
that holds because $\left[ \widehat{\Phi}^-(\x), \widehat{\Phi}^+(\y) \right]$ is a c-bidistribution times the identity and the fact that the right-hand side of eq. \eqref{NormOrdHad} is itself Hadamard. (See the proof of lemma \ref{LemmaComm} and {\it cf.} eq. \eqref{eqLem1}.)

Writing the initial state $\psi \in \mathscr{D}$ as in \eqref{lincomb}, we have that the operator $\widehat{L}$ acts as
\begin{align}
\widehat{L}: & \psi  = \sum_k c_k\ee^{\ii \widehat{\Phi}(f_k)^c} \Omega  \nonumber \\
& \mapsto \widehat{L} \psi  =  \sum_k c_k\widehat{L} \ee^{\ii \widehat{\Phi}(f_k)^c} \Omega.
\end{align}

Thus, we need only verify that expressions of the form\footnote{Actually it would suffice to verify this for $g=h$ since, by the methods of the proof of Proposition 6.1 in \cite{StrohmaierVerchWollenberg} (see also \cite{FewsterVerch} and \cite{Sanders} for related works) one may prove that, if $\psi_1$ and $\psi_2$ are Hadamard state vectors, which are quasi-free in the sense which allows also for a nonvanishing one-point function, then any linear combination is a Hadamard state vector.  (We thank Christopher Fewster for pointing this out to us.)}
\begin{equation}
\langle \widehat{L} \ee^{\ii \widehat{\Phi}(g)^c} \Omega  |  : \widehat{\Phi}(\x) \widehat{\Phi}(\y) : \widehat{L} \ee^{\ii \widehat{\Phi}(h)^c} \Omega \rangle/\langle \widehat{L} \psi  |  \, \widehat{L} \,  \psi \rangle
\end{equation}
are smooth, which can be done by showing that each of the expressions
\begin{subequations}
\begin{align}
\omega_2^{++}(\x, \y) & = \frac{\langle \widehat{L} \ee^{\ii \widehat{\Phi}(g)^c} \Omega  |   \widehat{\Phi}^+(\x) \widehat{\Phi}^+(\y) \, \widehat{L} \, \ee^{\ii \widehat{\Phi}(h)^c} \Omega \rangle}{\langle \widehat{L} \psi  |  \, \widehat{L} \,  \psi \rangle} \\
\omega_2^{+-}(\x, \y) & = \frac{\langle \widehat{L} \ee^{\ii \widehat{\Phi}(g)^c} \Omega  |  \widehat{\Phi}^+(\x) \widehat{\Phi}^-(\y) \, \widehat{L} \, \ee^{\ii \widehat{\Phi}(h)^c} \Omega \rangle}{\langle \widehat{L} \psi  |   \, \widehat{L} \,  \psi \rangle} \\
\omega_2^{--}(\x, \y) & = \frac{\langle \widehat{L} \ee^{\ii \widehat{\Phi}(g)^c} \Omega  |  \widehat{\Phi}^-(\x) \widehat{\Phi}^-(\y) \, \widehat{L} \, \ee^{\ii \widehat{\Phi}(h)^c} \Omega \rangle}{\langle \widehat{L} \psi  |   \, \widehat{L} \,  \psi \rangle}
\end{align}
\end{subequations}

\noindent is smooth. We shall show explicitly that $\omega_2^{+-}(\x, \y) \in C^\infty(M \times M)$. The rest of the calculations are similar. By the same techniques as before,
\begin{align}
\omega_2^{+-}(\x, \y) = \frac{\langle \Omega  | [\ee^{-\ii \widehat{\Phi}(g)} \widehat{L}, \widehat{\Phi}^+(\x)] [\widehat{\Phi}^-(\y),  \widehat{L}  \ee^{\ii \widehat{\Phi}(h)}] \Omega \rangle }{ \langle \widehat{L} \psi  |   \, \widehat{L} \,  \psi \rangle},
\end{align}
and, using the commutator relation $\left[\widehat{\Phi}^\pm(\x), \ee^{\ii \widehat{\Phi}(f)^c} \right] = \ii \left[\widehat{\Phi}^\pm(\x), \widehat{\Phi}^\mp(f)\right] \ee^{\ii \widehat{\Phi}(f)^c}$, we have that
\begin{align}
& \omega_2^{+-}(\x, \y)  = \left( \langle L \psi  |  \, \widehat{L} \,  \psi \rangle \right)^{-1} \nonumber \\ 
& \times \langle \Omega | \ee^{-\ii \widehat{\Phi}(g)^c} \left( \ii\left[\widehat{\Phi}^+(\x), \widehat{\Phi}^-(g) \right] \widehat{L} + \left[\widehat{L}, \widehat{\Phi}^+(\x) \right]  \right) \nonumber \\
& \times \left(-\left[\widehat{L}, \widehat{\Phi}^-(\y) \right] + \ii \left[\widehat{\Phi}^-(\y), \widehat{\Phi}^+(h) \right] \widehat{L} \right) \ee^{\ii \widehat{\Phi}(h)^c} \Omega \rangle.
\label{smooth+-}
\end{align}

By lemma \ref{LemmaComm} in appendix \ref{app:A1}, $\left[\widehat{\Phi}^+(\x), \widehat{\Phi}^-(g) \right]$ and $\left[\widehat{\Phi}^-(\y), \widehat{\Phi}^+(h) \right]$ define smooth functions and from condition {\it (i)} we conclude that $\omega_2^{+-}\in C^\infty(M\times M)$. A similar argument for $\omega_2^{++}$ and $\omega_2^{--}$ using conditions {\it (i)} and {\it (ii)} completes the proof.
\end{proof}

The next step in the proof is that our collapse operator $L_\z$ defined by eq. \eqref{Lz} satisfies the hypotheses of lemma \ref{Lem:Main1}. Indeed, conditions {\it (i)} and {\it (ii)} hold by the product rule of the commutator. For condition {\it (i)}, we have that
\begin{subequations}
\begin{align}
\left[\widehat{L}_\z, \widehat{\Phi}^\pm(f) \right]  & = \int_M d\vol(\x) \,f(\x) \, \left[\widehat{\Phi}^\mp(F), \widehat{\Phi}^\pm(\x) \right] \nonumber \\
& \times \left(-2 \alpha \left(\widehat{\Phi}(F) - Z \widehat{\I} \right) \right) \widehat{L}_\z, 
\label{ProdNot}
\end{align}
\end{subequations}
which is a smooth function times an operator by lemma \ref{LemmaComm} in appendix \ref{app:A1}.

For condition {\it (ii)} of lemma \ref{Lem:Main1}, we have
\begin{align}
& \left[\left[\widehat{L}_\z, \widehat{\Phi}^\pm(f) \right], \widehat{\Phi}^\pm(g)\right]  = \int_{M \times M} \!\!\!\!\!\! d\vol(\x) d\vol(\y) \,f(\x) \, g(\y)\nonumber \\
& \times  \left[\widehat{\Phi}^\mp(F), \widehat{\Phi}^\pm(\x) \right] \left[\widehat{\Phi}^\mp(F), \widehat{\Phi}^\pm(\y) \right] \nonumber \\
& \times  \left\{-2\alpha + \left[-2 \alpha  \left(\widehat{\Phi}(F) - Z \widehat{\I} \right) \right]^2 \right\}\widehat{L}_\z
\end{align}
which is a smooth bi-function times an operator by lemma \ref{LemmaComm} in appendix \ref{app:A1}. This concludes the proof of the first item of our theorem.

\vspace*{0.2cm}
\textbf{Proof of item } {\it 2}.
\vspace*{0.2cm}

We now show item {\it 2} of our theorem. Namely, that $\mathscr{D} \subset \mathscr{G} \subset \mathscr{H}$, where $\mathscr{G}$ contains vectors of the form
\begin{equation}
\psi = \sum_{n = 1}^N \psi_n = \sum_{n = 1}^N \prod_{k=0}^{n} \widehat{L}_{\z_k} \psi,
\end{equation}
where $\psi \in \mathscr{D}$.

To see that $\mathscr{D} \subset \mathscr{G}$, it suffices to notice that $\widehat{L}_z$ with $Z = 0$ and $\widehat{\Phi}(f) = 0$ ({\it e.g.} by demanding that the field is smeared against the zero function itself, or vanishes weakly as $\widehat{\Phi}((\Box - m^2 - \xi R) g) = 0$) is equal to the identity on the Hilbert space. Hence, all the vectors $\psi \in \mathscr{D}$ also belong to $\mathscr{G}$. This, in turn, guarantees that $\mathscr{G}$ is dense in $\mathscr{H}$.

We now show that $\psi \in \mathscr{G}$ defined, by
\begin{equation}
\psi = \sum_{n = 1}^N \alpha_n \psi_n = \sum_{n=1}^N \alpha_n \prod_{k=0}^{n} \widehat{L}_{\z_k} \psi,
\end{equation}
for $\psi \in \mathscr{D}$ and $\alpha_n \in \mathbb{C}$, such that $\langle \psi | \psi \rangle = 1$, is a Hadamard state vector. Once more, the strategy is to show that
\begin{align}
&\langle \psi | : \Phi(\x) \Phi(\y) : \psi \rangle \nonumber \\
&= \sum_{n=1}^N \sum_{m = 1}^N \overline{\alpha_n} \alpha_m \Big\langle \prod_{i=0}^{n} \widehat{L}_{\z_i} \psi \Big| : \widehat{\Phi}(\x) \widehat{\Phi} (\y) : \prod_{j = 0}^m \widehat{L}_{\z_j} \psi \Big\rangle
\end{align}
contributes smoothly to the two-point function. It suffices to show that each of the terms in
\begin{align}
\Big\langle \prod_{i=0}^{n} \widehat{L}_{\z_i} \psi \Big| \left(\widehat{\Phi}^+(\x) \widehat{\Phi}^+(\y) + \widehat{\Phi}^+(\x) \widehat{\Phi}^-(\y) \right. \nonumber \\
\left. + \widehat{\Phi}^+(\y) \widehat{\Phi}^-(\x) + \widehat{\Phi}^-(\x) \widehat{\Phi}^-(\y) \right) \prod_{j = 0}^m \widehat{L}_{\z_j} \psi \Big\rangle
\label{Multi1}
\end{align}
is a smooth bi-function.

We show this for the second term in \eqref{Multi1}. The rest of the calculations are similar. Let
\begin{equation}
\omega_{nm}^{+-}(\x, \y) = \Big\langle \prod_{i=0}^{n} \widehat{L}_{\z_i} \psi \Big| \widehat{\Phi}^+(\x) \widehat{\Phi}^-(\y) \prod_{j = 0}^m \widehat{L}_{\z_j} \psi \Big\rangle.
\end{equation}

We expand $\omega_{nm}^{+-}(\x, \y)$ as the sum of four terms, $\omega_{nm}^{+-}(\x, \y) = \omega_{nm}^{+-(1)}(\x, \y) + \omega_{nm}^{+-(2)}(\x, \y) + \omega_{nm}^{+-(3)}(\x, \y) + \omega_{nm}^{+-(4)}(\x, \y)$, where
\begin{subequations}
\begin{align}
\omega_{nm}^{+-(1)}(\x, \y) &  = \Big\langle \prod_{i=0}^{n} \widehat{L} _{\z_i} \widehat{\Phi}^-(\x) \psi \Big| \prod_{j = 0}^m \widehat{L}_{\z_j} \widehat{\Phi}^-(\y) \psi \Big\rangle \\
\omega_{nm}^{+-(2)}(\x, \y) &  = \Big\langle \prod_{i=0}^{n} \widehat{L} _{\z_i} \widehat{\Phi}^-(\x) \psi \Big|  \Big[\widehat{\Phi}^-(\y), \prod_{j = 0}^m \widehat{L}_{\z_j}\Big] \psi \Big\rangle \\
\omega_{nm}^{+-(3)}(\x, \y) &  = \Big\langle \Big[\widehat{\Phi}^-(\x), \prod_{i=0}^{n} \widehat{L} _{\z_i} \Big] \psi \Big|   \prod_{j = 0}^m \widehat{L}_{\z_j} \widehat{\Phi}^-(\y) \psi \Big\rangle \\
\omega_{nm}^{+-(4)}(\x, \y) &  = \Big\langle \Big[\widehat{\Phi}^-(\x), \prod_{i=0}^{n} \widehat{L} _{\z_i} \Big] \psi \Big|  \Big[\widehat{\Phi}^-(\y), \prod_{j = 0}^m \widehat{L}_{\z_j}\Big] \psi \Big\rangle
\end{align}
\label{4Terms}
\end{subequations}
and each of the commutators expand as
\begin{align}
& \Big[\widehat{\Phi}^-(\x), \prod_{i=0}^{n} \widehat{L} _{\z_i} \Big]  = \left[\widehat{\Phi}^-(\x), \widehat{L} _{\z_0} \right] \prod_{i=1}^{n} \widehat{L} _{\z_i} \nonumber \\
& + \sum_{i=1}^{n-1} \! \prod_{j=0}^{i-1} \! \widehat{L} _{\z_j} \! \left[\widehat{\Phi}^-(\x),  \widehat{L} _{\z_i} \right] \!\!\! \prod_{k=i+1}^{n} \widehat{L} _{\z_k} \! \!
 +  \!\! \prod_{i = 0}^{n-1} \!\! \widehat{L} _{\z_i} \left[\widehat{\Phi}^-(\x), \widehat{L} _{\z_n} \right].
\end{align}

At this stage, we know that each of the $\left[\widehat{\Phi}^-(\x), \widehat{L} _{\z_i} \right]$ is a smooth function times an operator, and following a strategy analogous to the proof of lemma \ref{Lem:Main1}, we conclude that each of the terms defined by the equations \eqref{4Terms} is smooth. From here, it follows that $\omega_{nm}^{+-} \in C^\infty(M \times M)$. A similar strategy shows that each of the terms on the right-hand side of eq. \eqref{Multi1} contributes smoothly and hence
$\langle \psi| : \widehat{\Phi}(\x) \widehat{\Phi}(\y) : \psi \rangle \in C^\infty(M \times M)$. Thus, we conclude that $\psi$ is a Hadamard state vector.

To complete the proof, we notice that for $\psi \in \mathscr{G}$, $\psi_\z = \widehat{L}_\z\psi /\langle \widehat{L}_\z \psi  \, | \, \widehat{L}_\z \psi \rangle^{1/2}$ is a normalized state in $\mathscr{G}$ and, hence, it is also Hadamard.

\end{proof}

\subsection{A perturbative result for higher order collapse generators}
\label{subsec:Pert}

In the previous subsection, our results dealt only with collapse operators with linear collapse generators.  In this subsection, we show how to deal with collapse operators whose collapse generators are higher order polynomials in a perturbative way. No claims will be made about the convergence of the pertubative expressions. The style of this subsection will be less rigorous; through a series of formal manipulations, we shall show that the post-collapse state has the Hadamard property order by order in a perturbative way. Thus, for the truncated perturbation series, the Hadamard condition is then recovered exactly. This is our second main result, which is summarized in theorem \ref{MainThm2}.

Let us suppose that we have a collapse generator given by a polynomial of degree $N$, $Q[\Phi, f_\z] = \sum_{i = 1}^N \prod_{j = 1}^i \Phi\left(P_g^{ij} f_\z\right)$, where the $P_g^{ij} f_\z$ are covariant expressions as before. For example, for $N = 3$,
\begin{align}
Q[\Phi, f_\z] & = \Phi(P_g^{11} f_\z) + \Phi(P_g^{21} f_\z)\Phi(P_g^{22} f_\z) \nonumber \\
& + \Phi(P_g^{31} f_\z) \Phi(P_g^{32} f_\z) \Phi(P_g^{33} f_\z).
\end{align}

It would be desireable to have a result such as the one stated in theorem \ref{MainThm}. The hypotheses of lemma \eqref{Lem:Main1}, however, need not hold due to convergence issues. Namely, the commutators in {\it (i)} and {\it (ii)} can be written only as formal series. But formal expressions are available for the commutators that produce two-point functions with the Hadamard ultraviolet behaviour, as we verify in this subsection. We make use of the following nested commutator notation:
\begin{equation}
\ad_X^n Y = \underbrace{[X, [X, \cdots [X, Y] \cdots]]}_{n \text{ commutators}},
\label{Nest}
\end{equation}
with the convention that $\ad_X^0 Y = Y$.

As before, we denote the GNS field representation of the Klein-Gordon field with respect to some Hadamard state by  $\widehat{\Phi}(f)= \pi(\Phi(f)) = (\ii\, \widehat{a}(KEF) - \ii\, \widehat{a}^*(KEf))$, as a linear operator acting on the (dense subset of the) Hilbert space $\mathscr{H}$. Once more, we denote the positive polarization of the field in this representation by $\widehat{\Phi}^+(f) = - \ii\, \widehat{a}^*(KEf)$ and the negative polarization by  $\widehat{\Phi}^-(f) = \ii\, \widehat{a}(KEf)$.

Let the collapse operator of the theory be $\widehat{L}_\z^N = \exp\left[-\alpha\Big(\widehat{Q}[\widehat{\Phi}, f_\z] - Z \widehat{\I} \Big)^2 \right]$. Then the commutators with the positive and negative polarizations of the field representation, $\widehat{\Phi}^+(f)$ and $\widehat{\Phi}^-(f)$ respectively, are 
\begin{widetext}
\begin{subequations}
\begin{align}
 \left[\widehat{L}_\z^N, \widehat{\Phi}^\pm(f) \right] &  = \int_M \! d\vol(\x) \, f(\x)  \left( \sum_{n = 1}^\infty \ad_{\left[-\alpha\left(\widehat{Q}\left[\widehat{\Phi}, f_\z \right] - Z \widehat{\I} \right)^2\right]}^n \widehat{\Phi}^\pm(\x) \right) \widehat{L}_\z^N, \label{PolComm1} \\
\left[\left[\widehat{L}_\z^N, \widehat{\Phi}^\pm(f) \right], \widehat{\Phi}^\pm(g)\right]  
& =  \int_M \!\!\!\! d\vol(\x) \, f(\x) g(\y) \left\{ \sum_{n = 1}^\infty \left[  \ad_{\left[-\alpha\left(\widehat{Q}\left[\widehat{\Phi}, f_\z \right] - Z \widehat{\I} \right)^2\right]}^n \widehat{\Phi}^\pm(\x), \widehat{\Phi}^\pm(\y) \right]  \right. \nonumber \\
& +   \left. \sum_{n = 1}^\infty \sum_{m = 1}^\infty \left( \ad_{\left[-\alpha\left(\widehat{Q}\left[\widehat{\Phi}, f_\z \right] - Z \widehat{\I} \right)^2\right]}^n \widehat{\Phi}^\pm(\x)\right)  \left( \ad_{\left[-\alpha\left(\widehat{Q}\left[\widehat{\Phi}, f_\z \right] - Z \widehat{\I} \right)^2\right]}^m \widehat{\Phi}^\pm(\y) \right)  \right\} \widehat{L}_\z^N.
\label{PolComm2}
\end{align}
\label{PolComm}
\end{subequations}
\end{widetext}
Eq. \eqref{PolComm1} follows from the ``adjoint-to-commutators" formal relation $\ee^X Y \ee^{-X} = \sum_{n = 0}^\infty \ad_X^n Y$, and hence the integrands appearing on the right-hand side of eq. \eqref{PolComm} should be understood as formal expressions. Still, we can show that these formal expressions contribute as smooth functions in the integrand on the right-hand side of \eqref{PolComm}. First, notice that, for each $n \in \mathbb{N}$, $\x \mapsto \ad_{\left[-\alpha\left(\widehat{Q}\left[\widehat{\Phi}, f_\z \right] - Z \widehat{\I} \right)^2\right]}^n  \widehat{\Phi}^\pm(\x)$ is an operator valued smooth function. This follows from the (formal) commutator
\begin{widetext}
\begin{align}
 \left[ -\alpha\left(\widehat{Q}\left[\widehat{\Phi}, f_\z \right] - Z \widehat{\I} \right)^2 , \widehat{\Phi}^\pm(\x) \right] & = -\alpha \sum_{i = 1}^N \sum_{j=1}^i \left[\widehat{\Phi}^\mp\left(P_g^{ij} f_\z\right) ,  \widehat{\Phi}^\pm(\x) \right] \nonumber \\
& \times  \left\{ \left(\widehat{Q}\left[\widehat{\Phi}, f_\z \right] - Z \widehat{\I} \right)  \prod_{k = 1, k \neq j}^i \widehat{\Phi}\left(P_g^{ik} f_\z\right)+ \prod_{k = 1, k \neq j}^i \widehat{\Phi}\left(P_g^{ik} f_\z\right) \left(\widehat{Q}\left[\widehat{\Phi}, f_\z \right] - Z \widehat{\I} \right)  \right\}
\label{SmoothPoly}
\end{align}
\end{widetext}
and lemma \ref{LemmaComm}, which establishes that $\left[\widehat{\Phi}^\mp\left(P_g^{ij} f_\z\right) ,  \widehat{\Phi}^\pm(\x) \right]$ is a smooth function times the identity operator. By the same argument, the second term on the right-hand side of eq. \eqref{PolComm2} contributes smoothly, and we need only verify the smoothness of the first term. To this end, we can use lemma \ref{LemmaInductiveComm} in appendix \ref{app:A} to write
\begin{widetext}
\begin{align}
\left[\left[\widehat{L}_\z^N, \widehat{\Phi}^\pm(f) \right], \widehat{\Phi}^\pm(g)\right]  & =  -\int_M \!\!\!\! d\vol(\x) \, f(\x) g(\y) \, \sum_{n = 0}^\infty  \nonumber \\
& \times \Bigg\{ \sum_{m=0}^{n-1} \! \ad^m_{\left[-\alpha\left(\widehat{Q}\left[\widehat{\Phi}, f_\z \right] - Z \widehat{\I} \right)^2\right]}  \left[ \left[\widehat{\Phi}^\pm(\y), -\alpha\left(\widehat{Q}\left[\widehat{\Phi}, f_\z \right] - Z \widehat{\I} \right)^2\right], \ad_{\left[-\alpha\left(\widehat{Q}\left[\widehat{\Phi}, f_\z \right] - Z \widehat{\I} \right)^2\right]}^{n-m-1} \widehat{\Phi}^\pm(\x)  \right]\nonumber \\
& - \left.  \sum_{m = 0}^\infty \left( \ad_{\left[-\alpha\left(\widehat{Q}\left[\widehat{\Phi}, f_\z \right] - Z \widehat{\I} \right)^2\right]}^n  \widehat{\Phi}^\pm(\x) \right) \left( \ad_{\left[-\alpha\left(\widehat{Q}\left[\widehat{\Phi}, f_\z \right] - Z \widehat{\I} \right)^2\right]}^m  \widehat{\Phi}^\pm(\y) \right)  \right\} \widehat{L}_\z^N.
\label{DoubleComm}
\end{align}
\end{widetext}

All the expressions on the right-hand side of eq. \eqref{DoubleComm} can be expanded, {\it cf.} \eqref{SmoothPoly}, into contributions that are smooth by lemma \ref{LemmaComm} in appendix \ref{app:A}. In particular, by our previous arguments, we see that the form of the commutators \eqref{PolComm} is
\begin{widetext}
\begin{subequations}
\begin{align}
  \left[\widehat{L}_\z^N, \widehat{\Phi}^\pm(f) \right]  & = \int_M \! d\vol(\x) \, f(\x) \, \sum_{i = 1}^N \sum_{j=1}^i \left[\widehat{\Phi}^\mp\left(P_g^{ij} f_\z\right) ,  \widehat{\Phi}^\pm(\x) \right] \left( \sum_{n = 1}^\infty R_{\alpha^n}^{ij}(\Phi, f_\z) \widehat{L}_\z^N \right),  \\
 \left[\left[\widehat{L}_\z^N, \widehat{\Phi}^\pm(f) \right], \widehat{\Phi}^\pm(g)\right]  
& =  \int_M \!\!\!\! d\vol(\x) \, f(\x) g(\y) \, \sum_{i = 1}^N \sum_{k = 1}^N \sum_{j=1}^i  \sum_{l=1}^k \left[\widehat{\Phi}^\mp\left(P_g^{ij} f_\z\right) ,  \widehat{\Phi}^\pm(\x) \right]  \nonumber \\
& \times  \left[\widehat{\Phi}^\mp\left(P_g^{kl} f_\z\right) ,  \widehat{\Phi}^\pm(\y) \right] \left( \sum_{n = 1}^ \infty \tilde{R}_{\alpha^n}^{ijkl}(\Phi, f_\z) \widehat{L}_\z^N \right).
\end{align}
\end{subequations}
\end{widetext}

Here, each one of the $R_{\alpha^n}^{ij}(\Phi, f_\z)$ and $\tilde{R}_{\alpha^n}^{ijkl}(\Phi, f_\z)$ can be calculated perturbatively by an expansion in the parameter $\alpha$. Thus, truncating the series at an arbitrary power of $\alpha$, and one sees that the singular behaviour of
\begin{equation}
\omega_2^f(\x,\y) = \frac{\langle \psi |  \widehat{L}_\z^N \widehat{\Phi}(\x) \widehat{\Phi}(\y) \widehat{L}_\z^N \psi_i \rangle}{ \langle \psi_i | \widehat{L}_\z^N \widehat{L}_\z^N \psi \rangle}
\end{equation}
is Hadamard to the prescribed power of $\alpha$ if $|\psi \rangle$ is an initial Hadamard state. We collect these observations in our second main theorem:

\begin{thm}
Let $(\pi, \cal{D} \subset \mathscr{H}, \Omega)$ be the GNS triple of the Klein-Gordon field theory with respect to an algebraic quasi-free Hadamard state, $\omega$, and let the domain $\mathscr D$ be defined as in Theorem \ref{MainThm}.  Let $\widehat{L}_\z^N\in \mathscr{L}(\mathscr{H})$ be a self-adjoint operator on the Hilbert space, defined as $\widehat{L}_\z^N = \exp\left[-\alpha\Big(\widehat{Q}[\widehat{\Phi}, f_\z] - Z \widehat{\I} \Big)^2 \right]$ with $Q[\Phi, f_\z] = \sum_{i = 1}^N \prod_{j = 1}^i \Phi\left(P_g^{ij} f_\z\right)$, for fixed $N \in \mathbb{N}$, $\alpha >0$, $Z \in \mathbb{R}$ and with $\widehat{\Phi}(P^{ij}_g f_\z)$ the representation of $\Phi(P^{ij}_g f_\z)$ (a covariantly smeared polynomial collapse generator in $f_\z$ in collapse model applications). Then, if $\psi \in {\rm Dom}\left(\widehat{L}_\z^N\right)$ is a Hadamard state, it follows that, for $\psi_\z^N = \widehat{L}_\z^N\psi/\langle  \widehat{L}_\z^N \psi \, | \, \widehat{L}_\z^N \psi \rangle^{1/2}$,
\begin{align}
\langle \psi_\z^N |\widehat{\Phi}(\x) \widehat{\Phi}(\y) \psi_\z^N \rangle - \langle \psi |\widehat{\Phi}(\x) \widehat{\Phi}(\y) \psi \rangle & = \sum_{n=0}^M \alpha^M G_M(\x, \y) \nonumber \\
& + O\left( \alpha^{M+1}\right)
\end{align}
for all $M \in \mathbb{N}$ and $G_M(\x,\y) \in C^\infty(M \times M)$.
 \qed
\label{MainThm2}
\end{thm}

\subsection{An example of the post-collapse renormalized energy-momentum tensor}
\label{subsec:Ex}

As an application of our results, we work out the renormalized energy momentum tensor in the post-collapse state $\psi_f  = \widehat{L}_\z\psi_i /\langle \widehat{L}_\z \psi_i  \, | \, \widehat{L}_\z \psi_i \rangle^{1/2}$, when the initial Hadamard state is a vacuum state of the theory $|\psi_i \rangle = | \Omega \rangle$ and in the simple case that $\widehat{L}_\z = \exp\left(-\alpha\left(\widehat{\Phi}(f_\z) - Z \widehat{\I}\right)^2 \right)$. The two-point function is
\begin{align}
\omega^f_2(\x,\y) & =  \omega^i_2(\x,\y) + \omega_f^{++}(\x,\y) + \omega_f^{+-}(\x,\y) \nonumber \\
& + \omega_f^{+-}(\y,\x) + \omega_f^{--}(\x,\y),
\end{align}
with $\omega^i_2(\x,\y) = \langle \Omega |[\widehat{\Phi}^-(\x),\widehat{\Phi}^+(\y)] \Omega \rangle$ and
\begin{widetext}
\begin{subequations}
\begin{align}
\omega_2^{f+-}&(\x,\y)  = \langle \Omega | \left[ \widehat{L}_\z, \widehat{\Phi}^+(\x) \right] \left[\widehat{\Phi}^-(\y), \widehat{L}_\z \right] \,\Omega \rangle / \langle \Omega | \widehat{L}_\z^2 \Omega \rangle \nonumber \\
& = \left[\widehat{\Phi}^-(f_\z),\widehat{\Phi}^+(\x) \right] \left[\widehat{\Phi}^+(f_\z),\widehat{\Phi}^-(\y) \right] \langle \psi_f | \left[-2\alpha \left(\widehat{\Phi}(f_\z) - Z\widehat{\I} \right) \right]^2 \psi_f \rangle, \\
\omega_2^{f++}&(\x,\y)  = \langle \Omega | \left[ \left[ \widehat{L}_\z, \widehat{\Phi}^+(\x) \right], \widehat{\Phi}^+(\y) \right] \widehat{L}_\z  \,\Omega \rangle / \langle \Omega | \widehat{L}_\z^2 \Omega \rangle \nonumber \\
& = \left[\widehat{\Phi}^-(f_\z),\widehat{\Phi}^+(\x) \right] \left[\widehat{\Phi}^-(f_\z),\widehat{\Phi}^+(\y) \right] \langle \psi_f | \left\{-2\alpha + \left[-2\alpha \left(\widehat{\Phi}(f_\z) - Z\widehat{\I} \right) \right]^2 \right\} \psi_f \rangle, \\
\omega_2^{f--}&(\x,\y)  = \langle \Omega | \widehat{L}_\z \left[ \left[ \widehat{L}_\z, \widehat{\Phi}^-(\y) \right], \widehat{\Phi}^-(\x) \right] \,\Omega \rangle / \langle \Omega | \widehat{L}_\z^2 \Omega \rangle \nonumber \\
& = \left[\widehat{\Phi}^+(f_\z),\widehat{\Phi}^-(\y) \right] \left[\widehat{\Phi}^+(f_\z),\widehat{\Phi}^-(\x) \right] \langle \psi_f | \left\{-2\alpha + \left[-2\alpha \left(\widehat{\Phi}(f_\z) - Z\widehat{\I} \right) \right]^2 \right\} \psi_f \rangle,
\end{align}
\end{subequations}
\end{widetext}
\noindent and collecting the terms we have that
\begin{widetext}
\begin{align}
\omega_2^f(\x,\y) & =  \langle \Omega |[\widehat{\Phi}^-(\x),\widehat{\Phi}^+(\y)] \Omega \rangle + \langle \psi_f |\left[-2\alpha \left(\widehat{\Phi}(f_\z) - Z\widehat{\I} \right) \right]^2 \psi_f \rangle \nonumber \\
& \times   \Big(\left[\widehat{\Phi}^-(f_\z),\widehat{\Phi}^+(\x) \right] \left[\widehat{\Phi}^-(f_\z),\widehat{\Phi}^+(\y) \right] + \left[\widehat{\Phi}^-(f_\z),\widehat{\Phi}^+(\x) \right] \left[\widehat{\Phi}^+(f_\z),\widehat{\Phi}^-(\y) \right] \nonumber \\
& + \left[\widehat{\Phi}^-(f_\z),\widehat{\Phi}^+(\y) \right] \left[\widehat{\Phi}^+(f_\z),\widehat{\Phi}^-(\x) \right] + \left[\widehat{\Phi}^+(f_\z),\widehat{\Phi}^-(\y) \right] \left[\widehat{\Phi}^+(f_\z),\widehat{\Phi}^-(\x) \right] \Big)\nonumber \\
& -2 \alpha \left(\left[\widehat{\Phi}^-(f_\z),\widehat{\Phi}^+(\x) \right] \left[\widehat{\Phi}^-(f_\z),\widehat{\Phi}^+(\y) \right] + \left[\widehat{\Phi}^+(f_\z),\widehat{\Phi}^-(\y) \right] \left[\widehat{\Phi}^+(f_\z),\widehat{\Phi}^-(\x) \right] \right).
\end{align}
\end{widetext}

Let $G^\pm \in C^\infty(M)$ be the smooth function defined by $G^\pm(\x) \widehat{\I} = \left[\widehat{\Phi}^\mp(f_\z),\widehat{\Phi}^\pm(\x) \right]$. $G^+$ and $G^-$ can be defined by integrating the Wightman function, $W(\x, \y) = \langle \Omega | \widehat{\Phi}(\y) \widehat{\Phi}(\x) \Omega \rangle$, as follows
\begin{subequations}
\begin{align}
G^+(\x) &= \langle \Omega | \widehat{\Phi}^-(f_\z) \widehat{\Phi}^+(\x) \Omega \rangle = \int_M \! d\vol(\y) \, f_\z(\y) W(\y, \x), \\
G^-(\x) &= \! -\langle \Omega | \widehat{\Phi}^-(\x) \widehat{\Phi}^+(f_\z)  \Omega \rangle \! = \! -\!\! \int_M \!\!\!\! d\vol(\y) \, f_\z(\y) W(\x, \y).
\end{align}
\end{subequations}

The post-collapse renormalized energy-momentum tensor is
\begin{widetext}
\begin{align}
 \langle \psi_f | T^{{\rm ren}}_{ab}  \psi_f \rangle & = \langle \Omega | T^{{\rm ren}}_{ab} \Omega \rangle  - 2 \alpha\Big\{\nabla_a G^+  \nabla_b G^+  + \nabla_a G^-  \nabla_b G^-  - \frac{1}{2}g_{ab} \left(\nabla_cG^+  \nabla^cG^+  + \nabla_cG^-  \nabla^cG^-  \right) \nonumber \\
&   - \frac{1}{2}g_{ab} \left(m^2 + \xi R  \right) \left( (G^+)^2 + (G^-)^2\right) \Big\} + 4\alpha^2 \langle \psi_f |\left(\widehat{\Phi}(f_\z) - Z\widehat{\I} \right)^2 \psi_f \rangle \Big\{\nabla_a G^+  \nabla_b G^+   \nonumber \\
&   + \nabla_{(a} G^+  \nabla_{b)} G^-  + \nabla_a G^-  \nabla_b G^-  - \frac{1}{2}g_{ab} \left(\nabla_cG^+  \nabla^cG^+  + 2\nabla^cG^+  \nabla^cG^-  + \nabla_cG^-  \nabla^cG^-  \right) \nonumber \\
&  - \frac{1}{2}g_{ab} \left(m^2 + \xi R  \right) \left((G^+)^2 + 2 G^+  G^-  + (G^-)^2\right) \Big\}.
\label{TrenEx}
\end{align}
\end{widetext}

Importantly, the in this example it is explicit that difference $\langle \psi_f | T^{{\rm ren}}_{ab}  \psi_f \rangle  - \langle \Omega | T^{{\rm ren}}_{ab} \Omega \rangle = O(\alpha)$, so for $\alpha \ll 1$, the change in the renormalized energy momentum tensor is small. In fact, this feature is general and from this standpoint one can begin to calculate back-reaction effects in semiclassical gravity, pertubatively in $\alpha$ is necessary.

%=====================================================================================================
% CONCLUDING REMARKS
%=====================================================================================================
\section{ Discussion and concluding remarks}
\label{sec:Conclusion}
 
Summarizing our  present work, we have presented a class of  operators  that  can be used  as  the operators  driving the  spontaneous collapse  dynamics  in the various generally covariant dynamical reduction models that generalize the GRW model. Further, we have proven that, for a wide class of Hadamard states for our model Klein-Gordon theory, they preserve the Hadamard property. In addition, we have worked out a simple example in which the violations of energy momentum are calculated and found to be small when the parameter $\alpha$ is small.

We have left out three important issues, which should be treated in future work: First, in our first main theorem \ref{MainThm}, we have shown that Hadamard vectors states belonging to certain dense subsets of the Hilbert space of the Klein-Gordon theory are mapped to Hadamard vector states by the effect of state reduction. As we have mentioned before, it remains an open question whether every state vector in our Hilbert space is mapped into a Hadamard state by each or any $\widehat{L}_\z$. Second, the inclusion of local polynomials in the admissible collapse generators, such as $\Phi^2(P_g f_\z)$, which entails dealing with the renormalization of the collapse generator itself. Third, the treatment of generally covariant generalizations of the more sophisticated CSL model, which has a much more `canonical' flavour, but for which a notion of {\it Hadamard on a slice} will need to first be developed.

Finally, the technology developed in this paper leaves us readily at the stage at which, given a suitable formulation of (an extended form of) semiclassical gravity that accounts for state reduction processes, one can compute back-reaction effects due to the state reduction on the spacetime.

A version of this work with further  remarks (and also an introduction to non-relativistic dynamical reduction models) is available in \cite{Juarez-Aubry:2017eryv1}.

%=====================================================================================================
% ACKNOWLEDGMENTS
%=====================================================================================================
\section*{Acknowledgments}

BAJ-A acknowledges the hospitality of the University of York and the University of Nottingham, where part of this work was carried out, as well as the support of an International Mobility Award granted by the Red Tem\'atica de F\'isica de Altas Energ\'ias (Red FAE) of the Consejo Nacional de Ciencia y Tecnolog\'ia, Mexico (CONACYT). BSK acknowledges the hospitality of ICN-UNAM, Mexico City, where part of this work was carried out.  BAJ-A thanks Daniel Bedingham for very helpful conversations in Oxford. BAJ-A and BSK thank Christopher Fewster for very helpful discussions in York. BSK thanks Umberto Lupo, and DS thanks Philip Pearle for very helpful correspondence. This work was supported by CONACYT project 101712, Mexico, and by PAPIIT- UNAM grant IG100316, Mexico.

\appendix
%=====================================================================================================
% APPENDIX
%=====================================================================================================
\section{Auxiliary results}
\label{app:A}

\subsection{A smoothness lemma}
\label{app:A1}

\begin{lemma}
Let $\mathscr{A}$ be the Klein-Gordon field algebra, $\omega: \mathscr{A} \to \mathbb{C}$ an algebraic Hadamard state and $(\pi, \mathscr{D} \subset \mathscr{H}, \Omega)$ be the GNS triple of the theory. Let $\pi(\Phi(f)) = \widehat{\Phi}^+(f) + \widehat{\Phi}^-(f)$, with $\widehat{\Phi}^+(f) = - \ii\, \widehat{a}^*(KEf)$ and $\widehat{\Phi}^-(f) = (\ii \, \widehat{a}(KEf)$. Then, $\x \mapsto \left[\widehat{\Phi}^\pm(\x), \widehat{\Phi}^\mp(f)\right]$ is equal to the identity operator times a smooth function on $M$.
\label{LemmaComm}
\end{lemma}
\begin{proof}
We shall complete the proof for $\left[\widehat{\Phi}^+(\x), \widehat{\Phi}^-(f)\right]$. The other case is analogous. That $\left[\widehat{\Phi}^+(\x), \widehat{\Phi}^-(f)\right]$ is a c-function times the identity is immediate because $\left[\widehat{\Phi}^+(g), \widehat{\Phi}^-(f) \right] = -\langle KEf, KEg \rangle_{\cal H} \widehat{\I}$, where $\langle \, , \, \rangle_{\cal H}$ denotes the one-particle Hilbert space inner product. (See e.g. \cite[app A]{Wald:1995yp} for details.) Thus, $\left[ \widehat{\Phi}^+(g), \widehat{\Phi}^-(f)\right] = \langle \Omega | \left[\widehat{\Phi}^+(g), \widehat{\Phi}^-(f)\right] \Omega \rangle \widehat{\I}$ and we have the normal ordering prescription
\begin{equation}
\widehat{\Phi}(g) \widehat{\Phi}(f) = \,\, : \widehat{\Phi}(g) \widehat{\Phi}(f) : + \langle \Omega | \left[\widehat{\Phi}^+(g), \widehat{\Phi}^-(f)\right] \Omega \rangle \widehat{\I},
\end{equation}
from where it follows that, 
\begin{align}
\left[\widehat{\Phi}^+(g), \widehat{\Phi}^-(f)\right] & = \left\{\lim_{\epsilon \to 0^+} \int_{M \times M} \!\!\!\!\!\!\!\!\!\!\!\! d\vol(\x) \, d\vol(\y) \, g(\x) f(\y) \, \right. \nonumber \\
& \left. \times \left[\chi(\x, \y) H_\epsilon^{T, n}(\x,\y) + W^n(\x,\y) \right]\right\} \widehat{\I},
\label{eqLem1}
\end{align}
where the right-hand side is as in def. \ref{Def:Hadamard}, {\it cf.} eq. \eqref{Had1}. As discussed in sec. \ref{sec:Hadamard}, below def. \ref{Def:Hadamard}, the $W^n$ can be seen to yield a smooth contribution and, hence, the proof is completed because
\begin{align}
S(\x) & = \frac{1}{(2 \pi)^2} \lim_{\epsilon \to 0^+} \int_{M} \! d\vol(\x) \, d\vol(\y) \, f(\y) \, \chi(\x, \y)  \nonumber \\
& \times \left(\frac{\Delta^{1/2}(\x, \y)}{\sigma_\epsilon(\x, \y)} + V^{(n)} \ln[\sigma_\epsilon(\x, \y) ] \right),
\end{align}
with the logarithm branch cut along the negative real axis, defines a smooth function on $M$, as has been shown in \cite[app B]{Kay:1988mu}.
\end{proof}

\subsection{A lemma for nested commutators}
\label{app:A2}

\begin{lemma}
Let $(\mathscr{L}, [ \, , \, ])$ be a Lie algebra. For $X, Y, Z \in \mathscr{L}$ and the adjoint notation ${\rm ad}_X^n Y$ defined as in eq. \eqref{Nest}, the following identity holds for all $n \in \mathbb{N}$:
\begin{align}
\left[ Y, \ad^n_X Z \right] & = \sum_{m = 0}^{n-1} \ad_X^{m} \Big( \left[ [Y,X], \ad_X^{n-m-1} Z \right] \Big) \nonumber \\
& + \ad_X^n ([Y,Z]).
\label{InductiveComm}
\end{align}

\label{LemmaInductiveComm}
\end{lemma}
\begin{proof}
We proceed by induction:

(i) For $n = 1$, the formula holds by Jacobi's identity.

(ii) We assume that \eqref{InductiveComm} holds for fixed $n$. For $n+1$,
\begin{align}
\left[ Y, \ad^{n+1}_X Z \right] & = \left[ Y, \left[X, \ad^{n}_X Z \right] \right] \nonumber \\
& = \left[ [Y, X] , \ad^{n}_X Z \right] + \left[ X, \left[Y, \ad^{n}_X Z \right] \right],
\end{align}
where in the second equality we used Jacobi's identity and by our hypothesis
\begin{widetext}
\begin{align}
 \left[ Y, \ad^{n+1}_X Z \right]  & = \left[ [Y, X] , \ad^{n}_X Z \right] + \left[ X, \left(\sum_{m = 0}^{n-1} \ad_X^{m} \Big( \left[ [Y,X], \ad_X^{n-m-1} Z \right] \Big) + \ad_X^n ([Y,Z]) \right)\right] \nonumber \\
& = \left[ [Y, X] , \ad^{n}_X Z \right] + \sum_{m = 0}^{n-1} \ad_X^{m+1} \Big( \left[ [Y,X], \ad_X^{n-m-1} Z \right] \Big) + \ad_X^{n+1} ([Y,Z]) \nonumber \\
& = \left[ [Y, X] , \ad^{n}_X Z \right] + \sum_{k = 1}^{(n+1)-1} \ad_X^{k} \Big( \left[ [Y,X], \ad_X^{(n+1)-k-1} Z \right] \Big) + \ad_X^{n+1} ([Y,Z])  \nonumber \\
& = \sum_{k = 0}^{(n+1)-1} \ad_X^{k} \Big( \left[ [Y,X], \ad_X^{(n+1)-k-1} Z \right] \Big) + \ad_X^{n+1} ([Y,Z]) ,
\end{align}
\end{widetext}
which completes the inductive step.

\end{proof}

%=====================================================================================================
% BIBLIOGRAPHY
%=====================================================================================================

\bibliography{HDRbib}
\end{document}